\documentclass[fleqn,english,british]{article}
\usepackage[T1]{fontenc}
\usepackage[latin9]{inputenc}
\usepackage{geometry}
\usepackage{color}
\usepackage{babel}
\usepackage{calc}
\usepackage{amsmath}
\usepackage{amsthm}
\usepackage{amssymb}
\usepackage{stmaryrd}
\usepackage[unicode=true,pdfusetitle,
 bookmarks=true,bookmarksnumbered=false,bookmarksopen=false,
 breaklinks=true,pdfborder={0 0 1},backref=false,colorlinks=false]
 {hyperref}

\makeatletter

\providecommand{\tabularnewline}{\\}

\numberwithin{equation}{section}
\numberwithin{figure}{section}
\theoremstyle{plain}
\newtheorem{thm}{\protect\theoremname}
\theoremstyle{definition}
\newtheorem{defn}[thm]{\protect\definitionname}
\theoremstyle{definition}
\newtheorem{example}[thm]{\protect\examplename}
\newenvironment{lyxlist}[1]
	{\begin{list}{}
		{\settowidth{\labelwidth}{#1}
		 \setlength{\leftmargin}{\labelwidth}
		 \addtolength{\leftmargin}{\labelsep}
		 }}
	{\end{list}}

\usepackage{syntax}
\usepackage{color}
\usepackage{xspace}

\definecolor{lightgray}{rgb}{.9,.9,.9}
\definecolor{darkgray}{rgb}{.4,.4,.4}
\definecolor{purple}{rgb}{0.65, 0.12, 0.82}

\AtBeginDocument{\catcode`\_=8}

\newcommand{\dd}[1]{\llbracket#1\rrbracket}
\newcommand{\ds}[1]{\llbracket\left|#1\right|_s\rrbracket}
\newcommand{\de}[1]{\llbracket\left|#1\right|_e\rrbracket}

\newcommand{\bind}{\Yright}

\newcommand\sdtl{\textsc{SDTL}\xspace}

\usepackage{hyperref}
\usepackage[all]{hypcap}

\makeatother

\usepackage{listings}
\addto\captionsbritish{\renewcommand{\definitionname}{Definition}}
\addto\captionsbritish{\renewcommand{\examplename}{Example}}
\addto\captionsbritish{}
\addto\captionsbritish{\renewcommand{\theoremname}{Theorem}}
\addto\captionsenglish{\renewcommand{\definitionname}{Definition}}
\addto\captionsenglish{\renewcommand{\examplename}{Example}}
\addto\captionsenglish{}
\addto\captionsenglish{\renewcommand{\theoremname}{Theorem}}
\providecommand{\definitionname}{Definition}
\providecommand{\examplename}{Example}
\providecommand{\theoremname}{Theorem}

\begin{document}

\title{Parametric Denotational Semantics for Extensible Language Definition
and Program Analysis}

\author{In-Ho Yi\\
Department of Computing and Information Systems \\
 The University of Melbourne, Victoria 3010, Australia \\
 i.yi@student.unimelb.edu.au}

\date{25th October 2012}

\maketitle
\global\long\def\st#1#2{\left|#1\right|_{s_{#2}}}

\global\long\def\ex#1#2{\left|#1\right|_{e_{#2}}}

\begin{abstract}
We present a novel approach to construction of a formal semantics
for a programming language. Our approach, using a parametric denotational
semantics, allows the semantics to be easily extended to support new
language features, and abstracted to define program analyses. We apply
this in analysing a duck-typed, reflective, curried dynamic language.
The benefits of this approach include its terseness and modularity,
and the ease with which one can gradually build language features
and analyses on top of a previous incarnation of a semantics.
\end{abstract}
\thispagestyle{empty}

\pagebreak{}

\setcounter{page}{1}
\pagenumbering{roman}

\tableofcontents{}

\pagebreak{}

\setcounter{page}{1}
\pagenumbering{arabic}

\section{Introduction}

Programming language semantics is a sub-field within theoretical computer
science where researchers develop formal descriptions for the meaning
of computer programs. Over the years, we have seen the development
of denotational semantics, where we mathematically model the effect
of an execution of a language construct. Operational semantics formalise
mechanical steps that transform program states given a particular
program. As we shall argue in this thesis, the challenge of analysing
dynamic languages, in both concrete and abstract manner, necessitates
a semantics that bridges the gap between the two different semantics
in order for such task to be feasible.

Abstract interpretation is a unifying theory for program analysis
and verification with which we ascertain run-time properties of a
program by approximating its semantics. The properties of interest
are almost always undecidable. The task of abstraction interpretation
can be thought of as over-approximating a set of concrete states in
a finite number of steps. The usual semantic domain is replaced by
an \emph{abstract} domain whose elements describe a set of run-time
states. Mathematically, such an abstract domain is a partially ordered
set (forming a lattice), the ordering corresponding to subset ordering
of the powerset of concrete states.

Two distinct needs motivated the development of the present work.
First, there is the theoretician's need for a simple, concise, elegant
way of presenting a formal semantics for a programming language and
of developing that into various static analyses. We base our approach
on a parametric denotational semantics that is modularised to allow
the concrete semantics and the abstract interpretation to share a
common framework that uniformly handles most aspects of the programming
language. Use of denotational semantics provides a strong foundation
in proving correctness of an abstract interpretation, and allows us
to focus on algorithmic details of analysis.

Second, there is a practical need for program analyses suitable for
the dynamic languages that have been growing in popularity in recent
years. Traditionally these languages were called ``scripting'' languages,
as they were mainly used for automating tasks and processing strings.
However, with the advent of web applications, languages such as Perl
and PHP gained popularity as languages for web application development.
On the client side, web pages make heavy use of JavaScript, a dynamically
typed language, to deliver dynamic contents to the browser. Recent
years have seen an increasing use of JavaScript on the server side,
as well.

What these languages provide is an ability to rapidly prototype and
validate application models in a real time read-eval-print loop. Another
strength comes from the fact that programmers do not need to have
a class structure defined upfront. Rather, class structures and types
of variables in general are dynamically built. This reduces the initial
overhead of software design.

However, these features come at a cost. The lack of a formal, static
definition of type information makes dynamically typed languages harder
to analyse. This difficulty causes several practical problems. 
\begin{itemize}
\item As applications become more mature, more effort is devoted to program
unit testing and writing assertions to ensure type safety of systems.
This extra effort can sometimes outweigh the benefit of having a dynamically
typed language.
\item Whereas programmers using statically typed languages enjoy an abundance
of development tools, the choice of tools for development in dynamically
typed languages is limited, and the tools that do exist lack much
of the power of the tools for statically typed languages, owing largely
to the difficulty or infeasibility of type analysis for such languages.
\item Lack of static type structure has a significant impact on the performance
of dynamically typed languages.
\end{itemize}
With these problems in mind, we have designed a model language that
has a dynamism comparable to that of the aforementioned scripting
languages, such as duck typing, reflection, and partial function application.
A notable omission is closure scoping. However, allowing function
currying gives expressive power to the language comparable to that
of languages with closure or lexical scoping.

The two concerns are not distinct ones, but an interconnected dialectic.
The theoretical need is there because of the difficulty of describing
the abstract and concrete meaning of dynamic languages, which often
allow side-effect causing, type-altering functions. With such complexity,
duck-typed languages are interesting test cases for which we formulate
concrete semantics, abstract interpretation and the proof of correctness.
Our Haskell implementation of both concrete and abstract analysis,
appearing in the appendix to this thesis, illustrates the practicality
of the proposed programming language semantics.

This work is inspired by Haskell's use of monads, and we assume the
reader's familiarity with monadic style Haskell programming. We also
assume knowledge of lambda notation, denotational semantics, order
and fixed point theory at the level of the textbook of Nielson and
Nielson's \cite{Nielson:1992:SAF:129085}.

In the following section, we discuss other works in the field of language
semantics and differentiate our work from them. In section \ref{sec:Overview},
we give a general overview of the proposed language semantics framework
and analysis. In sections \ref{sec:Analytic-framework} and \ref{sec:Semantic-functions},
we formally introduce our framework. In section \ref{sec:The-language-under},
we develop a model language with features gradually added on. We also
present concrete and abstract analysis of the language in each stage
of development in parallel. In sections \ref{sec:Well-definedness}
and \ref{sec:Correctness}, we argue formal properties of the language
analysis. Finally, in section \ref{sec:Conclusion-and-future}, we
conclude this thesis and discuss future direction.

\section{Related work}

Denotational semantics is the starting point of our development of
a formal framework. The idea of incorporating monads into denotational
definitions was developed by Liang and Hudak \cite{Liang:phd1998,Liang_Hudak:ESOP1996}.
Whereas these works modularise an analytic framework by having multiple
layers of monadic transformations, we instead parametrise the definition
of a program state. 

Action semantics, as advanced by Mosses \cite{Mosses1996TPA}, shares
the motivation that semantics ought to be pragmatic, yet expressive
enough to deal with non-trivial, feature-rich languages. While action
semantics endeavours to devise a new meta-language for describing
semantics, we constrain ourself to the language of denotational semantics,
and seek to devise a formalism largely compatible with denotational
semantics.

The idea of constructing formulae with parametric types can be found
in Wadler's work \cite{Wadler:1989:TF:99370.99404}. The present work
is a special application of the parametricity in the field of language
semantics and analysis.

Regarding the type analysis of dynamic languages, there have been
numerous studies \cite{Anderson:ECOOP2005,Anderson:ENTCS2005,Guha:ESOP2011,Wrigstad:POPL2010}
that consider simple toy languages and their semantics for the purpose
of static analysis of dynamic languages. A major difference between
those languages and the model language presented in this paper is
that our language is designed to capture the critical feature of real
world languages which allows functions to alter types through side-effect
causing statements. We point out similarities and differences of this
work compared to the cited works as we encounter them in this thesis.

Type analysis plays a crucial part in compiling scripting languages,
mainly to improve performance. Ancona et al \cite{Ancona:2007:RST:1297081.1297091}
and Dufour \cite{Dufour:2006} design restricted versions of scripting
languages so that static inference of types can be performed. We adopt
several techniques employed in those projects, such as the use of
named memory allocation sites as static references.

An important use case of functions in dynamically typed languages
is ``mixin'' functions \cite{Bracha:1990:MI:97946.97982}. By passing
arguments to a mixin function, objects can be extended with extra
methods; that is, functionality can be added dynamically. There are
model languages and formalisations of mixin functions, such as the
works of Anderson et al \cite{Anderson:ECOOP2005} and Mens et al
\cite{Mens:TR1996}. Where those works seek to find functional models
for mixins, we define instead a language (with side-effect causing
functions) that is expressive enough to program mixin inheritance.

Jensen et al \cite{Jensen:SAS2009} describe a feature-complete analyser
for the JavaScript language. Our work can be extended further to provide
the semantic foundation for such an analyser. Such an attempt to formalise
the analysis might pave the way for further refinement and improvement.

\section{\label{sec:Overview}Overview}

Our semantic framework is comprised of two components: one for the
syntactic structure, and the other for giving meanings to the primitive
operations. What divides the two is the following separation of concerns:
\begin{enumerate}
\item What are the semantic operations entailed in a particular syntactic
structure? For example, syntactic structure $\dd{a=30}$ entails a
primitive operation $\mathtt{asg}\left(a,30\right)$.
\item How do we interpret such semantic operations in a particular point
of view? If we were to give a concrete interpretation, we would interpret
$\mathtt{asg}\left(a,30\right)$ as updating an environment with a
newly defined variable e.g., $Env\left[a=30\right]$.
\end{enumerate}
Observe that an interpretation of syntactic structure can remain agnostic
of the structure of a program state at a given point. Therefore, once
we remove the actual interpretation of primitive operations, what
remains in a semantics can be re-used for multiple interpretations
of the language. Hence, not only are the primitive operations parametrised,
but so is the whole definition of the domain of the program state.
Such a separation of concerns also helps to define an extensible semantics,
to which adding a new feature takes as little effort as possible. 

Now we give a formal definition of our framework.
\begin{defn}[Parametric semantics]
 A parametric semantics $Q$ is a quintuple $\left(X,State,Value,s,P\right)$
where $X$ is a collection of semantic functions for syntactic structures,
as outlined below; $State$ is a set of representations of computation
state, which can be anything to suit a particular analysis; $Value$
is a set of all possible values that an expression can be evaluated
to be; $s$ is an initial program state; and $P$ is the set of primitive
operations of the semantics. We assume throughout that different states
are incomparable. In other words, $State$ is ordered by identity.
\end{defn}

Throughout the analysis, these primitive operations are the parameters
of our analysis:
\begin{itemize}
\item $\mathtt{esc}$ takes a state and reports whether it is escaping (i.e.,
whether or not control flow reaches the successor statement)
\item $\mathtt{cond}$ interprets the meaning of a branching point when
a value and two transformations (one for true and another for false)
are given
\item $\mathtt{asg}$ takes an identifier and a value, and performs assignment
\item $\mathtt{val}$ takes an identifier and produces its meaning
\item $\mathtt{conval}$ takes a constant and produces its meaning
\item $\mathtt{getinput}$ and $\mathtt{dooutput}$ define the meanings
of console I/O operations 
\item $\mathtt{bin}$ defines the meaning of all binary operations given
two values
\item $\mathtt{ret}$ defines the meaning of a return statement given the
value to be returned
\item $\mathtt{fundecl}$ defines the meaning of dynamic execution of a
function declaration
\item $\mathtt{apply}$ defines the meaning of (possibly partially) applying
a function to a list of values
\item $\mathtt{get}$ and $\mathtt{set}$ define the meaning of getting
or setting a member of an object
\item $\mathtt{getglobal}$ and $\mathtt{getthis}$ define the meanings
of keywords $\mathtt{global}$ and $\mathtt{this}$, respectively
\item $\mathtt{newobj}$ defines the meaning of instantiating a new object
from a particular allocation site
\end{itemize}
Types of these operations are given in section \ref{sec:The-language-under}
as we introduce them.

Our model language, as we let it evolve through this thesis, has a
set of features found commonly in scripting languages. In the remainder
of the thesis we provide the semantics for a language with many different
features. We introduce the components of the language step by step.
The aim is to demonstrate that the semantic formalism enables such
a stepwise development, each step being incremental in the sense that
it does not require revision of the semantic equations developed in
earlier steps. 

\begin{figure}[!t]
\begin{lstlisting}[frame=single]
function fact(f,x) {
	if(x < 2) { return 1; }
	return x * f(f,x-1);
}
output fact(fact,input);

fa=fact(fact);
output fa(input);

function Fruit(v) {
	this.value = v;
}

global.answer = 0;

function juicible(fruit, juice) {
	function juiceMe(j,x) {
		return this.value + j + x;
	}
	fruit.juice = juiceMe(juice); #currying
	global.answer=42;
}

# Juicibles
apple = new Fruit(15);
juicible(apple, 20);
grape = new Fruit(30);
juicible(grape, 50);

# Non-juicibles
banana = new Fruit(20);
watermelon = new Fruit(25);

output apple.juice(10); # 15 + 20 + 10
output grape.juice(10); # 30 + 50 + 10
output global.answer; # 42

try {
	if(input > 42) {throw 42;}
} catch(e) {
	output e;
}
\end{lstlisting}

\caption{\label{fig:Example-sdtl}Example \sdtl program}

\end{figure}

Figure \ref{fig:Example-sdtl} is an example of a program written
in the model language. We call this model language Simple Duck-Typed
Language(\sdtl). A locally-scoped procedural language with support
for higher order functions (lines 1 to 5) is introduced in Section
\ref{subsec:Procedural-language}. Function currying (lines 7, 8 and
20) is introduced in Section \ref{subsec:Function-currying}. Object
oriented features, including duck-typing and reflection, are introduced
in Section \ref{subsec:Object-Oriented}. Finally, exception handling
(lines 38 to 42) is introduced in Section \ref{subsec:Exception-handling}.

\section{\label{sec:Analytic-framework}Analytic framework}

In this section we introduce a monadic construct specifically designed
for the purpose of program analysis. We then introduce polymorphic
auxiliary functions that are useful in extending theories in a modular
manner.

First we define the monadic constructions. We define a type constructor
$M$ and a bind operator $\bind$. 
\begin{defn}[Type constructor]
 The type constructor $M$ has the following polymorphic definition.
$F$ is the set of program semantics. It is necessary to have this
as an input to the state transformation in order to give the fixed
point characterisation of semantics. $F$ is given a formal definition
in section \ref{sec:Semantic-functions}. The $a$ parameter to the
type is used in different context to extract different information
from the semantics.

\[
M\ a=F\rightarrow State\rightarrow\wp\left(State\times a\cup\left\{ \mathsf{Null}\right\} \right)
\]

Observe that a single state can give rise to multiple corresponding
successor states. We are essentially modelling a non-deterministic
state transformation. This gives us the flexibility to handle both
concrete and abstract semantics within a single framework.
\end{defn}

Every statement is understood as a state transformer. We distinguish
between ``normal'' and ``escaping'' statements, the latter yielding
an ``escape'' state. For example, when a function returns, the return
statement transforms the current state into an escape state. Our ``bind''
operator relies on a parametric operation $\mathtt{esc}$ to spell
out the precise mechanism for escaping the current program execution
flow. The $\mathtt{esc}$ function returns true if a state does not
continue to the next expression or statement (having encountered a
return statement, for example). This provides a flexible and general
formalisation of a control flow, and it allows the handling of exceptions
as well as function return statements.

In such escaping cases, there is no appropriate value of the type
$a$ to be associated with the successor states. Hence, we introduce
$\bot$ to be assigned to successor states of the escaping states.
\begin{defn}[Bind operator]
We define a bind operator $\bind$.

\begin{eqnarray*}
\left(\bind\right) & : & M\ a\rightarrow\left(a\rightarrow M\ b\right)\rightarrow M\ b\\
T\ \bind\ U & = & \lambda f\lambda s.\:\mathbf{let}\ S=T\ f\ s\ \mathbf{in}\underset{\left\langle s',a\right\rangle \in S}{\bigcup}\left(\begin{gathered}\mbox{if}\ \mathtt{esc}\left(s'\right)\ \mbox{then}\ \left\{ \left\langle s',\mathsf{Null}\right\rangle \right\} \ \mbox{else}\ U\ a\ f\ s'\end{gathered}
\right)
\end{eqnarray*}

\begin{defn}[Point-wise ordering of state transformations]
 Given $T,U\in State\rightarrow\wp\left(State\times a\cup\left\{ \mathsf{Null}\right\} \right)$,\\
$T\sqsubseteq U$ iff $\forall x\in State,\left(T\ x\right)\subseteq\left(U\ x\right)$

\begin{defn}[Point-wise ordering of monadic functions]
 Given $T,U\in M\ a$, $T\sqsubseteq U$ iff $\forall f\in F,\left(T\ f\right)\sqsubseteq\left(U\ f\right)$
\end{defn}

\end{defn}

\end{defn}

\begin{thm}[Preservation of monotonicity]
Given monads $T$, $T'$ and $U$, $T\sqsubseteq T'\implies T\bind U\sqsubseteq T'\bind U$ 
\end{thm}

\begin{proof}
When a state $x$ is a member of $T\bind U$ for some $f\in F$ and
an initial state $s$, there exists an intermediate state $x'\in\left(T\ f\ s\right)$
from which $x$ is derived by $U$. Clearly, such intermediate state
is also a member of $T'\ f\ s$ by definition of point-wise ordering. 

Formally, $\forall f\in F\forall s\in State,\exists x,x\in\left(\left(T\bind U\right)\ f\ s\right)\Leftrightarrow\exists x',x'\in\left(T\ f\ s\right)\wedge x\in\left(U\ f\ x'\right)$
by the definition of bind operation. Now, $T\sqsubseteq T'\implies x'\in\left(T'\ f\ s\right)$.
Hence, $x\in\left(\left(T'\bind U\right)\ f\ s\right)$.
\end{proof}
Having a monadic structure helps provide modularity. For example,
if a particular parametrised operation takes a state but only produces
a value, it would be redundant to include a state as a part of returning
type, to match the definition of monadic binding. In such a case,
we take a function that returns only a value, then lift it to be used
in the monadic context.
\begin{defn}[Monadic functions]
We define the following auxiliary functions to incorporate non-monadic
functions as a part of monadic transformation:

\begin{itemize}
\item (return for $M$) $I_{A}$ is an identity state transformer that takes
a constant and lifts it to an identity state transformer with the
constant as a return value

$I_{A}:a\rightarrow M\ a$

$I_{A}\ v=\lambda f\lambda s.\left\{ \left\langle s,v\right\rangle \right\} $
\item (lift for $M$) $I_{V}$ lifts a function that takes a state and returns
a value to a monadic function

$I_{V}:\left(State\rightarrow a\right)\rightarrow M\ a$

$I_{V}\ t=\lambda f\lambda s.\left\langle \left\{ \left\langle s,t\ s\right\rangle \right\} ,\emptyset\right\rangle $
\item $I_{S}$ takes a non-deterministic transformation and lifts it to
a monadic function

$I_{S}:\left(State\rightarrow\wp\left(State\right)\right)\rightarrow M\ a$

$I_{S}\ T=\lambda f\lambda s.\left\{ \left\langle s',\mathsf{Unit}\right\rangle \mid s'\in T\ s\right\} $

\end{itemize}
\begin{defn}[Record updater]
\label{recup} We model a state as a record with named fields. In
this way, an update operation written for a particular set of fields
can be reused without redefining it when we add extra dimensions to
a $State$ domain to accommodate features that are orthogonal to the
features of the previous version.

When we have a record $\rho$ with named fields $\left\langle c_{1},\ldots,c_{k}\right\rangle $,
and when an updater function $U$ updates fields $c_{i},\ldots,c_{j}$,
we define a function $\Upsilon$ that takes a record, projects its
fields into an n-tuple corresponding to the selected fields ($P_{c_{i},\ldots,c_{j}}$),
lets $U$ update the tuple, and finally updates the whole record with
the updated tuple ($J_{c_{i},\ldots,c_{j}}$).

\begin{eqnarray*}
P_{c_{i},\ldots,c_{j}} & = & \lambda\rho.\left\langle a_{i},\ldots,a_{j}\right\rangle \\
J_{c_{i},\ldots,c_{j}} & = & \lambda\left\langle a_{i},\ldots,a_{j}\right\rangle \lambda\rho.\rho\left[c_{n}=a_{n},n\in\left\{ i,\ldots,j\right\} \right]\\
\Upsilon_{c_{i},\ldots,c_{j}}\ U & = & J_{c_{i},\ldots,c_{j}}\circ U\circ P_{c_{i},\ldots,c_{j}}
\end{eqnarray*}
 where $a_{i}$ is a value for the field $c_{i}$ of a record $\rho$ 

Similarly, we define an operation to update a record and return a
value.

\[
\Psi_{c_{i},\ldots,c_{j}}\ U=\lambda\rho.\left\langle \rho',a\right\rangle 
\]

where $\left\langle \upsilon,a\right\rangle =U\ \left(P_{c_{i},\ldots,c_{j}}\ \rho\right)$
and $\rho'=J_{c_{i},\ldots,c_{j}}\ \upsilon$

Finally, we define a value extractor, that takes a record and selects
a value from it.

\[
\Theta{}_{c_{i},\ldots,c_{j}}\ U=U\circ P_{c_{i},\ldots,c_{j}}
\]

When $c$ is an n-tuple space for chosen fields and $r$ is a domain
of a $State$ record, the functions defined here have the following
type signatures:

\begin{eqnarray*}
\Upsilon & : & \left(c\rightarrow c\right)\rightarrow\left(r\rightarrow r\right)\\
\Psi & : & \left(c\rightarrow c\times a\right)\rightarrow\left(r\rightarrow r\times a\right)\\
\Theta & : & \left(c\rightarrow a\right)\rightarrow\left(r\rightarrow a\right)
\end{eqnarray*}
\end{defn}

\end{defn}

\begin{example}[Record updater example]
 To see these functions in use, suppose we have a simple record structure
for personal contacts.

\begin{eqnarray*}
Name & = & String\\
Address & = & String\\
Age & = & \mathbb{N}\\
Contact & = & Name\times Address\times Age
\end{eqnarray*}
\end{example}

\begin{itemize}
\item $\mathtt{addAge}$ updates age field of a contact record.
\end{itemize}
\[
\mathtt{addAge}\left(n\right)=\Upsilon_{Age}\ \lambda a.\left(a+n\right)
\]

\begin{itemize}
\item $\mathtt{getAgeAndAdd}$ returns the previous age field value while
updating the age field.
\end{itemize}
\[
\mathtt{getAgeAndAll}\left(n\right)=\Psi_{Age}\ \lambda a.\left\langle a+n,a\right\rangle 
\]

\begin{itemize}
\item $\mathtt{getAge}$ extracts age information from a contact record.
\end{itemize}
\[
\mathtt{getAge}=\Theta_{Age}\ \lambda a.a
\]

\begin{defn}[Singleton lifting]
Another commonly occurring pattern is that functions often return
a singleton set. We define a function that takes a function returning
a value and lifts it to be a function that returns a singleton set.

\begin{eqnarray*}
\Gamma & : & \left(a\rightarrow b\right)\rightarrow\left(a\rightarrow\wp\left(b\right)\right)\\
\Gamma\ F\ x & = & \left\{ F\ x\right\} 
\end{eqnarray*}

For simplicity of notation, we compose this function with the other
functions from Definition \ref{recup}.

\begin{eqnarray*}
\bar{\Upsilon} & = & \Gamma\circ\Upsilon\\
\bar{\Psi} & = & \Gamma\circ\Psi\\
\bar{\Theta} & = & \Gamma\circ\Theta
\end{eqnarray*}
\end{defn}

We now have monadic constructs and auxiliary functions to describe
the semantic functions of the model language. We can now define the
semantic functions of the language.

\section{\label{sec:Semantic-functions}Semantic functions}

We use syntax nodes as references to various items constituting program
environment. To all statements and expressions in a program, we designate
unique identifiers in order to reference them. For that purpose, we
define the following syntactic nodes and unique identifier spaces.
\begin{lyxlist}{00.00.0000}
\item [{$Stm$}] is the set of statement nodes.
\item [{$Exp$}] is the set of expression nodes.
\item [{$Lexp$}] is the set of left-expression nodes.
\item [{$Sid$}] is the set of statement identifiers.
\item [{$Eid$}] is the set of expression identifiers.
\item [{$Id$}] is the set of alphanumeric identifiers.
\end{lyxlist}
Note that we use an sid of a function declaration statement as a reference
point for the function defined. The $\mathtt{param}$ and $\mathtt{arity}$
functions take such an sid and return a list of parameter names, and
the arity of the function, respectively.

Where we specifically refer to an identifier to a syntactic construct,
we write $\dd{\left|X\right|_{x}}$ to mean a statement or expression
$X$ with an id $x$. In cases where such identifiers are not directly
referenced, we omit them for simplicity.
\begin{defn}[Semantic functions]
The analytic framework contains the following semantic functions:

\begin{eqnarray*}
F & = & Sid\rightarrow Stm\times\left(State\rightarrow\wp\left(State\right)\right)\\
\mathcal{S} & : & Stm\rightarrow F\rightarrow State\rightarrow\wp\left(State\times\left\{ \mathsf{Unit},\mathsf{Null}\right\} \right)\\
\mathcal{E} & : & Exp\rightarrow F\rightarrow State\rightarrow\wp\left(State\times Value\cup\left\{ \mathsf{Null}\right\} \right)\\
\mathcal{L} & : & Lexp\rightarrow F\rightarrow State\rightarrow\wp\left(State\times Value\cup\left\{ \mathsf{Null}\right\} \right)
\end{eqnarray*}

$\mathcal{S},\mathcal{E}$ and $\mathcal{L}$ are semantic functions
for statements, expressions and left expressions, respectively. $F$
is a function space to model the collection of functions in a program.
Given the sid of a function declaration site, it gives a statement
node for function declaration and a state transformer. Note that in
this picture a function ``returns'' a value by giving a state transformation.
Incorporating such a concept as a return value in a $State$ itself
provides a greater flexibility in describing the effects of executing
a statement or an expression at a particular program point.
\end{defn}

We define following auxiliary functions to describe the use of the
references to syntax nodes.

\begin{eqnarray*}
\mathtt{stm} & : & F\rightarrow Sid\rightarrow Stm\\
\mathtt{param} & : & F\rightarrow Sid\rightarrow\left[Id\right]\\
\mathtt{arity} & : & F\rightarrow Sid\rightarrow\mathbb{N}\cup\left\{ 0\right\} \\
\mathtt{run} & : & F\rightarrow Sid\rightarrow State\rightarrow\wp\left(State\right)
\end{eqnarray*}

\section{\label{sec:The-language-under}The language under study}

We define the model language, the \sdtl (Simple Duck-Typed Language).

\subsection{\label{subsec:Procedural-language}The procedural core language}

We start off with a procedural language with C-like syntax. 

\begin{grammar}

<con> ::= <Num> | <Bool>

<Lexp> ::= ID

<Exp> ::= <con> | <Lexp> | `input'
\alt  <Lexp> `(' [<Exp> [,<Exp>]*]? `)'
\alt  <Exp> <binop> <Exp>
\alt  `(' <Exp> `)'

<binop> ::= `+' | `-' | `*' | `/' | `>' | `<' | `=='

<Stm> ::= nil | <Stm> `;' <Stm> | <Exp>
\alt `output' <Exp>
\alt <Lexp> `=' <Exp>
\alt `if' `(' <Exp> `)' `{' <Stm> `}'
\alt `if' `(' <Exp> `)' `{' <Stm> `}' `else' `{' <Stm> `}'
\alt `while' `(' <Exp> `)' `{' <Stm> `}'
\alt `function' Id `(' [Id [, Id]*]? `)' `{' <Stm> `}'
\alt `return' <Exp>

\end{grammar}

Here \emph{Num} and \emph{Bool} are the syntactic categories for integers
and boolean values.

\sdtl does not have a separate category for function and variable
declarations. Variables are declared \emph{ad hoc} whenever such variables
appear as a left expression to assignment statements. Function declarations
are statements themselves, which allow them to appear anywhere in
the program.

\sdtl supports higher-order functions, which allows functions to
be recursively referenced. For example, we can define a factorial
function in a recursive manner.
\begin{example}[\label{Recursively-defined-factorial}Recursively defined factorial
function]
 In \sdtl, the factorial function can be implemented in a recursive
way.

\begin{lstlisting}
function fact(f,n) {
	if(n>1) { return f(f,n-1) * n; } else { return 1; }
}

z=fact(fact,input);
output z;
\end{lstlisting}

In this example, the function $\mathtt{fact}$ takes two arguments.
The first is the function pointer to recursively invoke, and the second
is the usual argument to the function. This example illustrates that
recursive functions are possible even in the absence of lexical scoping
or other special scoping rules to allow a function body to refer to
the function itself.
\end{example}

Given the availability of higher-order functions, we formulate the
meaning of a function as a fixed point (see the definition of $F$).
The semantic functions for \sdtl are defined in figure \ref{semantic}.
We use auxiliary functions $\mathtt{evalParams}$ and $\mathtt{call}$
to describe a function call. (Note that we use $\epsilon$ for the
empty sequence, $\left[a\right]$ for the set of sequences of any
number of values of type $a$, and the notation $a\smallfrown b$
to denote concatenation of sequences $a$ and $b$.)

\begin{eqnarray*}
\mathtt{call} & : & Sid\rightarrow\left[Value\right]\rightarrow M\ Value\\
\mathtt{call}\left(n,p\right) & = & \lambda f\lambda\rho.\left\{ \mathtt{leave}\ \rho\ \rho'\mid\rho'\in S\right\} \ \mbox{where}\ S=\mathtt{run}\left(f,n\right)\ \left(\mathtt{enter}\ \rho\ n\ p\ \mathtt{param}\left(f,n\right)\right)
\end{eqnarray*}

\begin{eqnarray*}
\mathtt{evalParams} & : & \left[Exp\right]\times\left[Value\right]\rightarrow F\rightarrow State\rightarrow\wp\left(State\times\left[Value\right]\right)\\
\mathtt{evalParams}\ \epsilon\ ps\ f\ \rho & = & \lambda f\lambda\rho.\left\{ \left\langle \rho,ps\right\rangle \right\} \\
\mathtt{evalParams}\ \dd{E}\smallfrown Exps\ ps & = & \lambda f\lambda\rho.\underset{\left\langle \rho',con\right\rangle \in S}{\bigcup}\left\{ \mathtt{evalParams}\ Exps\ \left(ps\smallfrown con\right)\ f\ \rho'\right\} \\
 &  & \mbox{where}\;S\ =\ \mathcal{E}\ \dd{E}\ f\ \rho
\end{eqnarray*}

$\mathtt{enter}$ is a parametrised function that takes a caller's
state at the time of function invocation and an id-to-constant value
mapping, and constructs an initial state for a callee. $\mathtt{leave}$
takes both the caller's state and the resulting states of callee's,
and constructs the caller's states after the function call. These
functions are parametrised so as to allow each interpretation to define
the exact shape of a program state and its manipulation during a function
call and return. These functions have the following types:

\begin{eqnarray*}
\mathtt{enter} & : & State\rightarrow Sid\rightarrow\left[Value\right]\rightarrow\left[Id\right]\rightarrow State\\
\mathtt{leave} & : & State\rightarrow State\rightarrow State
\end{eqnarray*}

\begin{figure*}[!t]
\noindent\fbox{\begin{minipage}[t]{1\textwidth - 2\fboxsep - 2\fboxrule}%
\begin{align*}
 & \mathcal{F} & = & Y\left(\lambda f.\lambda n.\left\langle \mathtt{stm}\left(f,n\right),\mathcal{S}\ \dd{\mathtt{stm}\left(f,n\right)}\ f\right\rangle \right)\\
 & \mathcal{S}\dd{\epsilon} & = & I_{A}\bot\\
 & \mathcal{S}\dd{S_{1};S_{2}} & = & \begin{aligned}\mathcal{S}\ \dd{S_{1}}\bind\lambda\_.\mathcal{S}\ \dd{S_{2}}\end{aligned}
\\
 & \mathcal{S}\dd{E} & = & \mathcal{E}\ \dd{E}\bind\lambda\_.I_{A}\mathsf{Unit}\\
 & \mathcal{S}\dd{\mathtt{return}\ E} & = & \mathcal{E}\ \dd{E}\bind\lambda v.I_{S}\ \mathtt{ret}\left(V\right)\\
 & \mathcal{S}\dd{\mathtt{if}\left(E\right)\ S_{1}} & = & \mathcal{E}\ \dd{E}\bind\lambda v.\mathtt{cond}\left(v,\mathcal{S}\ \dd{S_{1}},I_{A}\mathsf{Unit}\right)\\
 & \mathcal{S}\dd{\mathtt{if}\left(E\right)\ S_{1}\ \mathtt{else}\ S_{2}} & = & \mathcal{E}\ \dd{E}\bind\lambda v.\mathtt{cond}\left(v,\mathcal{S}\ \dd{S_{1}},\mathcal{S}\ \dd{S_{2}}\right)\\
 & \mathcal{S}\dd{id=E} & = & \begin{aligned}\mathcal{E}\ \dd{E} & \bind\lambda v.I_{S}\ \mathtt{asg}\left(id,v\right)\end{aligned}
\\
 & \mathcal{S}\dd{\mathtt{while}\left(E\right)\ S_{1}} & = & Y\lambda x.\lambda f\lambda r.\left(\begin{aligned}\mathcal{E}\ \dd{E}\bind\lambda v.\mathtt{cond}\left(\begin{gathered}v,\left(\mathcal{S}\ \dd{S_{1}}\right)\bind\lambda\_.x,I_{A}\mathsf{Unit}\end{gathered}
\right)\end{aligned}
\right)\\
 & \mathcal{S}\dd{\mathtt{output}\ E} & = & \mathcal{E}\ \dd{E}\bind\lambda v.I_{S}\ \mathtt{dooutput}\left(v\right)\\
 & S\ds{\mathtt{function}\ Id_{1}\left(\vec{Id}\right)\ S_{1}} & = & I_{S}\ \mathtt{fundecl}\left(Id_{1},s\right)\\
 & \mathcal{E}\dd{con} & = & I_{A}\ \mathtt{conval}\left(con\right)\\
 & \mathcal{E}\dd{L} & = & \mathcal{L}\ \dd{L}\\
 & \mathcal{E}\dd{\mathtt{input}} & = & \mathtt{getinput}\\
 & \mathcal{E}\dd{L\left(\vec{E}\right)} & = & \mathcal{L}\ \dd L\bind\lambda n.\left(\mathtt{evalParams}\ \vec{E}\ \emptyset\right)\bind\lambda p.\mathtt{call}\left(n,p\right)\\
 & \mathcal{E}\dd{E_{1}\ binop\ E_{2}} & = & \mathcal{E}\ \dd{E_{1}}\bind\lambda c_{1}.\mathcal{E}\ \dd{E_{2}}\bind\lambda c_{2}.I_{V}\ \mathtt{bin}\left(binop,c_{1},c_{2}\right)\\
 & \mathcal{L}\dd{Id} & = & I_{A}\ \mathtt{val}\left(Id\right)
\end{align*}
\end{minipage}}

\label{semantic}\caption{Semantic equations for a procedural core of \sdtl}
\end{figure*}

The types of primitive operations are as follows.

\begin{eqnarray*}
\mathtt{ret} & : & Value\rightarrow State\rightarrow\wp\left(State\right)\\
\mathtt{cond} & : & Value\rightarrow M\ Value\rightarrow M\ Value\rightarrow M\ Value\\
\mathtt{asg} & : & Id\rightarrow Value\rightarrow State\rightarrow\wp\left(State\right)\\
\mathtt{dooutput} & : & Value\rightarrow State\rightarrow\wp\left(State\right)\\
\mathtt{fundecl} & : & Id\rightarrow Sid\rightarrow State\rightarrow\wp\left(State\right)\\
\mathtt{conval} & : & \mathbb{Z}\cup\left\{ \mathsf{true},\mathsf{false}\right\} \rightarrow Value\\
\mathtt{getinput} & : & M\ Value\\
\mathtt{bin} & : & \left\{ +,-,*,/\right\} \rightarrow\mathbb{Z}\cup\left\{ \mathsf{true},\mathsf{false}\right\} \rightarrow\mathbb{Z}\cup\left\{ \mathsf{true},\mathsf{false}\right\} \rightarrow Value\\
\mathtt{val} & : & Id\rightarrow State\rightarrow Value
\end{eqnarray*}

\subsubsection{Concrete interpretation}

\paragraph{Domain}

\begin{eqnarray*}
FunPointer & = & Sid\\
Value & = & \mathbb{Z}\cup\left\{ \mathsf{true},\mathsf{false}\right\} \cup\:FunPointer\\
Env & = & Id\rightarrow Value\\
Return & = & Value\cup\left\{ \mathsf{Void}\right\} \\
CState & = & Env\times IO\times Return
\end{eqnarray*}

$CState$ is a Cartesian product of environment, input/output state
and a return value. The return value is set to be a value when a function
is returning any value inside a function body. This has been incorporated
as a part of a program state so that we can signal escaping from a
program flow. The initial program state is $\left\langle \emptyset,io,\mathsf{Void}\right\rangle ,$
where $io$ is an initial IO state.

\paragraph{Functions}

\begin{eqnarray*}
\mathtt{enter} & = & \lambda\left\langle \_,IO,\_\right\rangle ,n,P,p.\left\langle \left[p_{k}\mapsto P_{k}\right],IO,\mathsf{Void}\right\rangle \\
\mathtt{leave} & = & \lambda\left\langle V,\_,\_\right\rangle ,\left\langle \_,IO',R\right\rangle .\left\langle \left\langle V,IO',\mathsf{Void}\right\rangle ,R\right\rangle 
\end{eqnarray*}

\begin{eqnarray*}
\mathtt{esc} & = & \bar{\Theta}_{Return}\ \left(\lambda R.\left(R\neq\mathsf{Void}\right)\right)\\
\mathtt{cond}\left(v,s_{1},s_{2}\right) & = & \begin{cases}
s_{1} & \mbox{if}\ v=\mathsf{true}\\
s_{2} & \mbox{if}\ v=\mathsf{\mathsf{false}}
\end{cases}\\
\mathtt{asg}\left(id,v\right) & = & \bar{\Upsilon}_{Env}\ \left(\lambda V.V\left[id\mapsto v\right]\right)\\
\mathtt{val}\left(id\right) & = & \bar{\Upsilon}_{Env}\ \left(\lambda V.V\left(id\right)\right)\\
\mathtt{conval}\left(con\right) & = & \left\{ con\right\} \\
\mathtt{getinput\left(\rho\right)} & = & \bar{\Psi}_{IO}\ \lambda O.\left\langle O',\mbox{user input integer}\right\rangle \ \mbox{see note}\\
\mathtt{dooutput}\left(v\right) & = & \bar{\Upsilon}_{IO}\ \lambda O.O'\ \mbox{see note}\\
\mathtt{fundecl}(id,n) & = & \bar{\Upsilon}_{Env}\ \left(\lambda V.V\left[id\mapsto n\right]\right)\\
\mathtt{ret}\left(v\right) & = & \bar{\Upsilon}_{Return}\ \left(\lambda R.v\right)\\
\mathtt{bin}\left(op,c_{1},c_{2}\right) & = & c_{1}\ op\ c_{2}\ \mbox{(perform binary operation between two constants)}
\end{eqnarray*}

We omit a detailed description of the IO environment. Normally, IO
can be modelled as a queue of inputs and outputs as they are given
and produced during the execution of a program.
\begin{example}[Concrete interpretation of recursive factorial function]
The program in example \ref{Recursively-defined-factorial} is concretely
interpreted as follows.

\begin{itemize}
\item At line 1, $\mathtt{fundecl}\left(\mathtt{fact},1\right)$ updates
environment to be $\left[\mathtt{fact}\mapsto\mathsf{Function}\ 1\right]$
assuming the function declaration has a unique id of 1.
\item At line 5, $\mathtt{getinput}$ gives a user input. Assume that the
input was 2. $\mathtt{evalParam}$ gives $\left[\left(\mathsf{Function}\ 1\right),2\right]$.
In a function call $\mathtt{call}\left(1,\left[\left(\mathsf{Function}\ 1\right),2\right]\right)$,
we first construct the initial state of a function call. $\mathtt{enter}$
gives $\left\langle \left[\mathtt{f}\mapsto\left(\mathsf{Function}\ 1\right),\mathtt{n}\mapsto2\right],IO,Void\right\rangle $.
\item At line 2, evaluating expression $\mathtt{n}>1$ yields true. $\mathtt{cond}$
invokes another function call, with initial state $\left\langle \left[\mathtt{f}\mapsto\left(\mathsf{Function}\ 1\right),\mathtt{n}\mapsto1\right],IO,Void\right\rangle $.
\item On the second call \texttt{fact(f,1)}, $\mathtt{n}>1$ yields false.
Hence, $\mathtt{cond}$ invokes $\mathtt{ret}\left(1\right)$ which
gives final state of $\left\langle \left[\mathtt{f}\mapsto\left(\mathsf{Function}\ 1\right),\mathtt{n}\mapsto1\right],IO,1\right\rangle $.
\item On the first call, this state is first evaluated to yield value $1$
by $\mathtt{leave}$. Then, \texttt{f(f,1) {*} 2} evaluates to 2,
which becomes the ultimate return value.
\item At line 5, after the function call $\mathtt{leave}$ gives $\left\langle \left[\mathtt{fact}\mapsto\mathsf{Function}\ 1\right],IO,\mathsf{Void}\right\rangle $
as the final state. $\mathtt{asg}$ adds symbol $\mathtt{z}$ to the
environment: $\mathtt{z}\mapsto2$
\item At line 6, $\mathtt{val\left(z\right)}$ evaluates to 2, which is
the final output of the program.
\end{itemize}
\end{example}

\subsubsection{Abstract interpretation}

At this stage, abstract interpretation looks largely similar to concrete
interpretation. Notable differences are that we approximate each constant
by its type, and that $\mathtt{cond}$ is a non-deterministic transformation
where it collects effects of both branches at a branching point.

\paragraph{Domain}

\begin{eqnarray*}
FunPointer & = & Sid\\
AVal & = & \left\{ \mathsf{Num},\mathsf{Bool}\right\} \cup FunPointer\\
AEnv & = & Id\rightarrow AVal\\
AReturn & = & AVal\cup\left\{ \mathsf{Void}\right\} \\
AState & = & AEnv\times AReturn
\end{eqnarray*}

Composition of an abstract domain is similar to that of the concrete
counterpart, except that it does not include an IO state. $\bot$
is an undetermined value, which is used to approximate unknown function
calls at an initial stage. The initial program state is $\left\langle \emptyset,\mathsf{Void}\right\rangle $.

\paragraph{Functions}

\begin{eqnarray*}
\mathtt{enter} & = & \lambda\left\langle \_,\_\right\rangle ,n,P,p.\left\langle \left[p_{k}\mapsto P_{k}\right],\mathsf{Void}\right\rangle \\
\mathtt{leave} & = & \lambda\left\langle V,\_\right\rangle ,\left\langle \_,R\right\rangle .\left\langle \left\langle V,\mathsf{Void}\right\rangle ,R\right\rangle 
\end{eqnarray*}

\begin{eqnarray*}
\mathtt{esc} & = & \bar{\Theta}_{AReturn}\ \left(\lambda R.\left(R\neq\mathsf{Void}\right)\right)\\
\mathtt{cond}\left(v,s_{1},s_{2}\right) & = & \lambda f\lambda\eta.\left(s_{1}\ f\ \eta\right)\cup\left(s_{2}\ f\ \eta\right)\\
\mathtt{asg}\left(id,v\right) & = & \bar{\Upsilon}_{AEnv}\ \left(\lambda\sigma.\sigma\left[id=v\right]\right)\\
\mathtt{val}\left(id,v\right) & = & \bar{\Upsilon}_{AEnv}\ \left(\lambda\sigma.\sigma\left(id\right)\right)\\
\mathtt{conval}\left(con\right) & = & \begin{cases}
\Gamma\ \mathsf{Num} & \mbox{if}\ con\in\mathbb{N}\\
\Gamma\ \mathsf{Bool} & \mbox{if}\ con\in\mathbb{B}
\end{cases}\\
\mathtt{getinput} & = & \Gamma\ \left(\lambda\eta.\left\langle \eta,\mathsf{Num}\right\rangle \right)\\
\mathtt{dooutput}\left(v\right) & = & \Gamma\ \left(\lambda\eta.\eta\right)\\
\mathtt{fundecl}\left(id,n\right) & = & \bar{\Upsilon}_{AEnv}\ \left(\lambda\sigma.\sigma\left[id\mapsto n\right]\right)\\
\mathtt{ret}\left(v\right) & = & \bar{\Upsilon}_{AReturn}\ \left(\lambda R:AReturn.v\right)\\
\mathtt{bin}\left(op,c_{1},c_{2}\right) & = & \begin{cases}
\Gamma\ \mathsf{Num} & \mbox{if}\ op\in{\scriptstyle \left\{ +,-,*,/\right\} }\\
\Gamma\ \mathsf{Bool} & \mbox{otherwise}
\end{cases}
\end{eqnarray*}

Function definitions are largely similar to that of concrete definition.
\begin{example}[\label{abs-int-fact}Abstract interpretation of recursive factorial
function]
The program in example \ref{Recursively-defined-factorial} is abstractly
interpreted as follows.

\begin{itemize}
\item At line 1, $\mathtt{fundecl}\left(\mathtt{fact},1\right)$ updates
environment to be $\left[\mathtt{fact}\mapsto\mathsf{Function}\ 1\right]$
assuming the function declaration has a unique id of 1.
\item At line 5, $\mathtt{getinput}$ gives an abstract value $\mathsf{\mathsf{Num}}$.
$\mathtt{evalParam}$ gives $\left[\left(\mathsf{Function}\ 1\right),\mathsf{\mathsf{Num}}\right]$.
In a function call $\mathtt{call}\left(1,\left[\left(\mathsf{Function}\ 1\right),\mathsf{\mathsf{Num}}\right]\right)$,
we first construct initial state of a function call. $\mathtt{enter}$
gives 
\[
\left\langle \left[\mathtt{f}\mapsto\left(\mathsf{Function}\ 1\right),\mathtt{n}\mapsto\mathsf{Num}\right],IO,\mathsf{Void}\right\rangle 
\]
\item The meaning of this function call is determined via a fixed point
iteration by progressively updating the current approximation of the
meaning of the function call, starting from a null hypothesis that
the function call does not return any state.
\end{itemize}
\begin{tabular}{|c|c|c|}
\hline 
Current approximation & Meaning of function call & Note\tabularnewline
\hline 
\hline 
$\emptyset$ & $\left\{ \left\langle \left[\mathtt{fact}\mapsto\left(\mathsf{Function}\ 1\right)\right],\mathsf{\mathsf{Num}}\right\rangle \right\} $ & \tabularnewline
\hline 
$\left\{ \left\langle \left[\mathtt{fact}\mapsto\left(\mathsf{Function}\ 1\right)\right],\mathsf{\mathsf{Num}}\right\rangle \right\} $ & $\left\{ \left\langle \left[\mathtt{fact}\mapsto\left(\mathsf{Function}\ 1\right)\right],\mathsf{\mathsf{Num}}\right\rangle \right\} $ & Fixed point\tabularnewline
\hline 
\end{tabular}

Implementation of this fixed point iteration is found in the $\mathtt{fun}$
function in appendix \ref{subsec:Abstract-interpretation}.

\begin{itemize}
\item This yields $\left\langle \left[\mathtt{fact}\mapsto\left(\mathsf{Function}\ 1\right)\right],\mathsf{\mathsf{Num}}\right\rangle $
\item After $\mathtt{leave}$,$\mathtt{asg}$ and $\mathtt{dooutput}$,
we have the final state of the program: $\left[\mathtt{fact}\mapsto\mathsf{Function}\ 1,\mathtt{z}\mapsto\mathsf{\mathsf{Num}}\right]$
\end{itemize}
\begin{example}[Abstract interpretation of a while loop]
 The following example illustrates an interpretation of a while loop
through a fixed point iteration.

The program calculates sum of a sequence. For the purpose of illustration,
we have added variable $\mathtt{x}$ that changes its type inside
a while loop.

\begin{lstlisting}
sum = 0;
z = input;
x = 50;
while(z>0) {
	sum = sum + z;
	z = z - 1;
	x = true;
}

output sum;
\end{lstlisting}
\end{example}

\begin{itemize}
\item At line 4, we have environment $\left[\mathtt{sum}\mapsto\mathsf{\mathsf{Num}},\mathtt{z}\mapsto\mathsf{\mathsf{Num}},\mathtt{x}\mapsto\mathsf{\mathsf{Num}}\right]$.
\item As an initial hypothesis, we assume that the statement body of a while
loop does not cause any change in the program state for any given
initial state. We progressively update this approximation until we
meet a fixed point.
\end{itemize}
\begin{tabular}{|c|c|}
\hline 
Current Approximation & Init $\rightarrow$ Final State\tabularnewline
\hline 
\hline 
$\emptyset$ & $\left[\begin{array}{ccc}
\mathtt{sum} & \mapsto & \mathsf{Num}\\
\mathtt{z} & \mapsto & \mathsf{Num}\\
\mathtt{x} & \mapsto & \mathsf{Num}
\end{array}\right]\rightarrow\left\{ \begin{gathered}\left[\begin{array}{ccc}
\mathtt{sum} & \mapsto & \mathsf{Num}\\
\mathtt{z} & \mapsto & \mathsf{Num}\\
\mathtt{x} & \mapsto & \mathsf{Num}
\end{array}\right]\\
\left[\begin{array}{ccc}
\mathtt{sum} & \mapsto & \mathsf{Num}\\
\mathtt{z} & \mapsto & \mathsf{Num}\\
\mathtt{x} & \mapsto & \mathsf{Bool}
\end{array}\right]
\end{gathered}
\right\} $\tabularnewline
\hline 
$\left[\begin{array}{ccc}
\mathtt{sum} & \mapsto & \mathsf{Num}\\
\mathtt{z} & \mapsto & \mathsf{Num}\\
\mathtt{x} & \mapsto & \mathsf{Num}
\end{array}\right]\rightarrow\left\{ \begin{gathered}\left[\begin{array}{ccc}
\mathtt{sum} & \mapsto & \mathsf{Num}\\
\mathtt{z} & \mapsto & \mathsf{Num}\\
\mathtt{x} & \mapsto & \mathsf{Num}
\end{array}\right]\\
\left[\begin{array}{ccc}
\mathtt{sum} & \mapsto & \mathsf{Num}\\
\mathtt{z} & \mapsto & \mathsf{Num}\\
\mathtt{x} & \mapsto & \mathsf{Bool}
\end{array}\right]
\end{gathered}
\right\} $ & $\begin{gathered}\left[\begin{array}{ccc}
\mathtt{sum} & \mapsto & \mathsf{Num}\\
\mathtt{z} & \mapsto & \mathsf{Num}\\
\mathtt{x} & \mapsto & \mathsf{Num}
\end{array}\right]\rightarrow\left\{ \begin{gathered}\left[\begin{array}{ccc}
\mathtt{sum} & \mapsto & \mathsf{Num}\\
\mathtt{z} & \mapsto & \mathsf{Num}\\
\mathtt{x} & \mapsto & \mathsf{Num}
\end{array}\right]\\
\left[\begin{array}{ccc}
\mathtt{sum} & \mapsto & \mathsf{Num}\\
\mathtt{z} & \mapsto & \mathsf{Num}\\
\mathtt{x} & \mapsto & \mathsf{Bool}
\end{array}\right]
\end{gathered}
\right\} \\
\left[\begin{array}{ccc}
\mathtt{sum} & \mapsto & \mathsf{Num}\\
\mathtt{z} & \mapsto & \mathsf{Num}\\
\mathtt{x} & \mapsto & \mathsf{Bool}
\end{array}\right]\rightarrow\left[\begin{array}{ccc}
\mathtt{sum} & \mapsto & \mathsf{Num}\\
\mathtt{z} & \mapsto & \mathsf{Num}\\
\mathtt{x} & \mapsto & \mathsf{Bool}
\end{array}\right]
\end{gathered}
$\tabularnewline
\hline 
$\begin{gathered}\left[\begin{array}{ccc}
\mathtt{sum} & \mapsto & \mathsf{Num}\\
\mathtt{z} & \mapsto & \mathsf{Num}\\
\mathtt{x} & \mapsto & \mathsf{Num}
\end{array}\right]\rightarrow\left\{ \begin{gathered}\left[\begin{array}{ccc}
\mathtt{sum} & \mapsto & \mathsf{Num}\\
\mathtt{z} & \mapsto & \mathsf{Num}\\
\mathtt{x} & \mapsto & \mathsf{Num}
\end{array}\right]\\
\left[\begin{array}{ccc}
\mathtt{sum} & \mapsto & \mathsf{Num}\\
\mathtt{z} & \mapsto & \mathsf{Num}\\
\mathtt{x} & \mapsto & \mathsf{Bool}
\end{array}\right]
\end{gathered}
\right\} \\
\left[\begin{array}{ccc}
\mathtt{sum} & \mapsto & \mathsf{Num}\\
\mathtt{z} & \mapsto & \mathsf{Num}\\
\mathtt{x} & \mapsto & \mathsf{Bool}
\end{array}\right]\rightarrow\left[\begin{array}{ccc}
\mathtt{sum} & \mapsto & \mathsf{Num}\\
\mathtt{z} & \mapsto & \mathsf{Num}\\
\mathtt{x} & \mapsto & \mathsf{Bool}
\end{array}\right]
\end{gathered}
$ & $\begin{gathered}\left[\begin{array}{ccc}
\mathtt{sum} & \mapsto & \mathsf{Num}\\
\mathtt{z} & \mapsto & \mathsf{Num}\\
\mathtt{x} & \mapsto & \mathsf{Num}
\end{array}\right]\rightarrow\left\{ \begin{gathered}\left[\begin{array}{ccc}
\mathtt{sum} & \mapsto & \mathsf{Num}\\
\mathtt{z} & \mapsto & \mathsf{Num}\\
\mathtt{x} & \mapsto & \mathsf{Num}
\end{array}\right]\\
\left[\begin{array}{ccc}
\mathtt{sum} & \mapsto & \mathsf{Num}\\
\mathtt{z} & \mapsto & \mathsf{Num}\\
\mathtt{x} & \mapsto & \mathsf{Bool}
\end{array}\right]
\end{gathered}
\right\} \\
\left[\begin{array}{ccc}
\mathtt{sum} & \mapsto & \mathsf{Num}\\
\mathtt{z} & \mapsto & \mathsf{Num}\\
\mathtt{x} & \mapsto & \mathsf{Bool}
\end{array}\right]\rightarrow\left[\begin{array}{ccc}
\mathtt{sum} & \mapsto & \mathsf{Num}\\
\mathtt{z} & \mapsto & \mathsf{Num}\\
\mathtt{x} & \mapsto & \mathsf{Bool}
\end{array}\right]
\end{gathered}
$\tabularnewline
\hline 
\end{tabular}

Implementation of this fixed point iteration can be found in the$\mathtt{fix}$
function in appendix \ref{subsec:Abstract-interpretation}.

\begin{itemize}
\item The resulting final states of the program is calculated to be $\left[\begin{array}{ccc}
\mathtt{sum} & \mapsto & \mathsf{Num}\\
\mathtt{z} & \mapsto & \mathsf{Num}\\
\mathtt{x} & \mapsto & \mathsf{Num}
\end{array}\right],\left[\begin{array}{ccc}
\mathtt{sum} & \mapsto & \mathsf{Num}\\
\mathtt{z} & \mapsto & \mathsf{Num}\\
\mathtt{x} & \mapsto & \mathsf{Bool}
\end{array}\right]$
\end{itemize}
\end{example}

\subsection{\label{subsec:Function-currying}Function currying}

We now introduce function currying to the \sdtl language. Introduction
of this language feature allows the language to be flexible enough
to express what JavaScript programmers would do with lexical scoping.
\begin{example}[\label{Function-currying}Function currying]
 We take a simple add function, and curry one argument to produce
different adders.

\begin{lstlisting}
function add(x,y) {
	return x+y;
}

add5 = add(5);
add7 = add(7);

output add5(input) + add7(input);
\end{lstlisting}

If we were to write this in JavaScript, we could have written the
following for the same effect.

\begin{lstlisting}
function adder(toadd) {
	return function(y) {
		return toadd+y;
	}
}

add5 = adder(5);
add7 = adder(7);
// suppose input is a platform-specific console input function
console.log(add5(input()) + add7(input()));
\end{lstlisting}
\end{example}

The introduction of function currying does not change the syntax of
the language. Therefore, there is no inherent reason for changing
semantic functions. However, we do redefine the semantics of function
calls to include an eid as an input, for the reason explained below.

\[
\mathcal{E}\de{L\left(\vec{E}\right)}=\mathcal{L}\ \dd{L}\bind\lambda n.\mathtt{evalParams}\ \vec{E}\ \emptyset\bind\lambda p.\mathtt{apply}\left(n,p,e\right)
\]

Here we introduce the $\mathtt{apply}$ function. It invokes the $\mathtt{call}$
function when the function arguments are saturated, or it returns
a pointer to a curried function otherwise. Its type is as follows.

\begin{eqnarray*}
\mathtt{apply} & : & Value\rightarrow\left[Value\right]\rightarrow Eid\rightarrow M\ Value
\end{eqnarray*}

\subsubsection{Concrete interpretation}

We need to extend the definition of $FunPointer$ to hold curried
parameters. This means that the $\mathtt{fundecl}$ function also
needs to be modified to match the new type signature.

\paragraph{Domain}

$FunPointer=Sid\times\left[Value\right]$

The initial program state is unchanged.

\paragraph{Functions}

\begin{eqnarray*}
\mathtt{apply}\left(\left\langle n,C\right\rangle ,p,\_\right) & =\lambda f. & \begin{cases}
\mathtt{call}\left(n,C\smallfrown p\right)\ f & \mbox{if}\begin{gathered}\:\left\Vert C\smallfrown p\right\Vert =\mathtt{arity}\left(f,n\right)\end{gathered}
\\
\lambda\rho.\left\{ \left\langle \rho,\left\langle n,C\smallfrown p\right\rangle \right\rangle \right\} \  & \mbox{otherwise}
\end{cases}\\
\mathtt{fundecl}(id,n) & = & \bar{\Upsilon}_{Env}\ \lambda V.V\left[id\mapsto\left\langle n,\emptyset\right\rangle \right]
\end{eqnarray*}

Here, $\left\Vert s\right\Vert $ is the length of a sequence $s$.
\begin{example}[Concrete interpretation of a function currying]
 Consider the program shown in example \ref{Function-currying}.
\end{example}

\begin{itemize}
\item At line 1, we are presented with a function declaration. $\mathtt{fundecl}$
adds the identifier as a reference to a function pointer with no curried
value. Assuming $add$ has given a unique id of $1$ during the parsing
of the program, we have in environment $\left[\mathtt{add}\mapsto\left\langle 1,\emptyset\right\rangle \right]$.
\item At lines 5 and 7, we partially apply the $\mathtt{add}$ function.
Since the number of arguments is not saturated, $\mathtt{apply}\left(\left\langle 1,\emptyset\right\rangle ,\left[5\right],\_\right)$
gives $\left\langle 1,\left[5\right]\right\rangle $ for $\mathtt{add5}$.
Similarly, $\mathtt{add7}$ gets $\left\langle 1,\left[7\right]\right\rangle $.
\item At line 8, we saturate the parameters, $\mathtt{apply}\left(\left\langle 1,\left[5\right]\right\rangle ,\left[x\right],\_\right)$
give $\mathtt{call}\left(1,\left[5,x\right]\right)$, where $x$ is
an arbitrary number given from user input. Hence we have the addition
done. It works similarly for $\mathtt{add7}$.
\end{itemize}

\subsubsection{Abstract interpretation}

Note that curried functions introduce a possibility of creating closures
requiring an infinite number of arguments.
\begin{example}[\label{Currying-loop}Currying loop]
 Consider the following program. 
\end{example}

\begin{lstlisting}
function foo(a,b) {
	return a;
}

x = 0;
while(true) {
	x = foo(x);
}
\end{lstlisting}

If we naively interpret this program, we would not be able to reach
a fixed point in analysis. Instead, we would have:

\begin{eqnarray*}
\mathtt{x} & = & \mathtt{foo}\left(Num\right)\\
\mathtt{x} & = & \mathtt{foo}\left(\mathtt{foo}\left(Num\right)\right)\\
\mathtt{x} & = & \mathtt{foo}\left(\mathtt{foo}\left(\mathtt{foo}\left(Num\right)\right)\right)\\
 & \vdots
\end{eqnarray*}

A solution to this problem is to have a curried function anchored
to a particular language construct. In this case, we can use an eid
of a curried expression as a point of reference (or '0' if not curried).

\paragraph{Domain}

\begin{eqnarray*}
FunPointer & = & Sid\times c\times\left(Eid\cup\left\{ 0\right\} \right)\\
Curried & = & Sid\times c\times Eid\rightarrow\left[AVal\right]\ \mbox{where}\ c\mbox{ is a number of parameters given}\\
AState & = & AEnv\times Curried\times AReturn
\end{eqnarray*}

$Curried$ takes a function id, number of argument curried and the
eid of a calling site, and gives the list of curried values. The initial
program state is $\left\langle \emptyset,\emptyset,\mathsf{Void}\right\rangle $.

\paragraph{Functions}

Definitions are given in figure \ref{absfun-curry}

\begin{figure*}[!t]
\noindent\fbox{\begin{minipage}[t]{1\textwidth - 2\fboxsep - 2\fboxrule}%
\begin{align*}
\mathtt{apply}\left(\left\langle n,c,e\right\rangle ,p,e'\right) & =\begin{cases}
\Theta_{{\scriptscriptstyle Curried}}\ \lambda\nu.\underset{C\in\nu\left(n,c,e\right)}{\bigcup}\mathtt{call}\left(n,C\smallfrown p\right)\ \mbox{if }c+\left\Vert p\right\Vert =\mathtt{{arity}\left(n\right)}\\
\Psi_{{\scriptscriptstyle Curried}}\ \lambda\nu.\underset{C\in\nu\left(n,c,e\right)}{\bigcup}\left\langle \nu\left[\left\langle n,c+\left\Vert p\right\Vert ,e'\right\rangle \mapsto\left(C\smallfrown p\right)\right],\left\langle n,c+\left\Vert p\right\Vert ,e'\right\rangle \right\rangle \ {\scriptstyle \mbox{otherwise}}
\end{cases}\\
\mathtt{fundecl}\left(id,n\right) & =\bar{\Upsilon}_{AEnv}\ \lambda\sigma.\sigma\left[id\mapsto\left\langle n,0,0\right\rangle \right]
\end{align*}
\end{minipage}}

\caption{Abstract functions for curried functions}
\label{absfun-curry}
\end{figure*}

\begin{example}[Abstract interpretation of a currying loop]
 The program shown in example \ref{Currying-loop} can be analysed
as follows:
\end{example}

\begin{itemize}
\item At line 1, $\mathtt{fundecl}\left(\mathtt{foo},1\right)$ gives $\left[\mathtt{foo}\mapsto\left\langle 1,0,0\right\rangle \right]$,
assuming $1$ is a unique id given to the function. 
\item At line 7, assuming the partial application expression has a unique
id of 7, $\mathtt{apply}\left(\left\langle 1,0,0\right\rangle ,\left[\mathsf{Num}\right],7\right)$
gives $\left[x\mapsto\left\langle 1,1,7\right\rangle \right]$ for
environment and $\left[\left\langle 1,1,7\right\rangle \mapsto\left[\mathsf{Num}\right]\right]$
for curried in the first iteration. Subsequently, we have $\mathtt{apply}\left(\left\langle 1,0,0\right\rangle ,\left\langle 1,1,7\right\rangle ,7\right)$
gives $\left[x\mapsto\left\langle 1,1,7\right\rangle \right]$ for
environment and $\left[\left\langle 1,1,7\right\rangle \mapsto\left[\left\langle 1,1,7\right\rangle \right]\right]$
for the list of curried value. Here we reach a fixed point.
\item All possible final states calculated by the interpretation are:

\begin{tabular}{|c|c|}
\hline 
Environment & Curried\tabularnewline
\hline 
\hline 
$x\mapsto\mathsf{Num}$ & None\tabularnewline
\hline 
$x\mapsto\left\langle 1,1,7\right\rangle $ & $\left\langle 1,1,7\right\rangle \mapsto\left[\mathsf{Num}\right]$\tabularnewline
\hline 
$x\mapsto\left\langle 1,1,7\right\rangle $ & $\left\langle 1,1,7\right\rangle \mapsto\left[\left\langle 1,1,7\right\rangle \right]$\tabularnewline
\hline 
\end{tabular}
\end{itemize}

\subsection{\label{subsec:Object-Oriented}Object oriented features}

We now extend our language to support objects.

\begin{grammar}

<Lexp> ::= <Exp> `.' ID

<Exp> ::= `global' | `this' | `new' <Lexp> `(' [<Exp> [,<Exp>]*]? `)'

\end{grammar}

$\mathtt{global}$ is a reference to a global object, and the reference
to global object is program invariant (akin to `window' in Client-side
JavaScript or `global' in Node.js). $\mathtt{this}$ is the usual
reference to the receiver object of a method call. 

As this extension introduces new syntactic structure, extra definitions
for semantic functions are given in figure \ref{semmon-oo}. Note,
however, that \emph{existing semantic equations are not impacted by
this}.

\begin{figure*}[!t]
\noindent\fbox{\begin{minipage}[t]{1\textwidth - 2\fboxsep - 2\fboxrule}%
\begin{align*}
 & \mathcal{S}\dd{E_{1}.id=E_{2}} & = & \mathcal{E}\ \dd{E_{1}}\bind\lambda r.\mathcal{E}\ \dd{E_{2}}\bind\lambda v.I_{S}\ \mathtt{set}\left(r,id,v\right)\\
 & \mathcal{E}\dd{\mathtt{global}} & = & I_{V}\ \mathtt{getglobal}\\
 & \mathcal{E}\dd{\mathtt{this}} & = & I_{V}\ \mathtt{getthis}\\
 & \mathcal{E}\de{L\left(\vec{E}\right)} & = & \left(\begin{aligned}\mathcal{L}\ \dd L & \bind\lambda n.\\
\left(\mathtt{evalParams}\ \vec{E}\ \emptyset\right) & \bind\lambda p.\\
I_{V}\ \mathtt{getthis} & \bind\lambda t.\\
\mathtt{apply}\left(n,p,t,e\right)
\end{aligned}
\right)\\
 & \mathcal{E}\de{\mathtt{new}\ L\left(\vec{E}\right)} & = & \left(\begin{aligned}\mathcal{L}\ \dd{L} & \bind\lambda n.\\
\mathtt{evalParams}\ \vec{E}\ \emptyset & \bind\lambda p.\\
\mathtt{newobj}\left(e\right) & \bind\lambda m.\\
\mathtt{apply}\left(n,p,m,e\right) & \bind\lambda\_.\\
I_{A}\ m
\end{aligned}
\right)\\
 & \mathcal{E}\de{E_{1}.id\left(\vec{E}\right)} & = & \left(\begin{aligned}\mathcal{E}\ \dd{E_{1}} & \bind\lambda t.\\
I_{V}\ \mathtt{get}(t,id) & \bind\lambda n.\\
\mathtt{evalParams}\ \vec{E}\ \emptyset & \bind\lambda p.\\
\mathtt{apply}\left(n,p,t,e\right)
\end{aligned}
\right)\\
 & \mathcal{L}\dd{E_{1}.id} & = & \mathcal{E}\ \dd{E_{1}}\bind\lambda v.I_{V}\ \mathtt{get}\left(v,id\right)
\end{align*}
\end{minipage}}

\caption{Semantic equations for the object-oriented extension of the language}
\label{semmon-oo}
\end{figure*}

We redefine the $\mathtt{call}$ function to include receiver object
reference, and to account for side-effects that functions can have
on object memories.

\selectlanguage{english}%
\begin{eqnarray*}
\mathtt{call} & : & Sid\rightarrow\left[Value\right]\rightarrow Value\rightarrow M\ Value\\
\mathtt{call}\left(n,p,t\right) & = & \lambda f\lambda\rho,\gamma.\Theta\ \left\langle \left\{ \mathtt{leave}\ \rho\ \rho'\mid\rho'\in S\right\} ,\gamma'\right\rangle \ \mbox{where}\left\langle S,\gamma'\right\rangle =f\ n\ \left(\mathtt{enter}\ \rho\ n\ p\ t\ \mathtt{param}\left(f,n\right)\right)\ \gamma
\end{eqnarray*}

\selectlanguage{british}%
\begin{eqnarray*}
\mathtt{enter} & : & State\rightarrow Sid\rightarrow\left[Value\right]\rightarrow Value\rightarrow\left[Id\right]\rightarrow State\\
\mathtt{leavel} & : & State\rightarrow State\rightarrow State
\end{eqnarray*}

Note that by modelling function calls as an effect on a state, we
can readily model such concepts as mixins. In the $\mbox{JS}_{0}$
language \cite{Anderson:ECOOP2005,Anderson:ENTCS2005}, mixin creating
behaviour is modelled in a functional language. A similar approach
has been taken in the discussion of the $\lambda_{S}$ language by
Guha et al \cite{Guha:ESOP2011}. However, with languages like JavaScript,
it is in the nature of such a language that functions are side-effect
causing, and much of the type operations are being done by the side-effects.
Hence our deviation from a purely functional approach to devise a
model language that resembles a real-life language.

\subsubsection{Concrete interpretation}

We extend the domain to include object memory (where object references
are mapped to a symbol mapping), and a reference to $\mathtt{this}$
object.

\paragraph{Domain}

\begin{eqnarray*}
Value & = & \mathbb{Z}\cup\left\{ \mathsf{true},\mathsf{false}\right\} \cup Object\cup FunPointer\\
Object & = & \mathbb{N}\\
ObjMem & = & n\rightarrow Env\\
This & = & Object\\
CState & = & Env\times ObjMem\times This\times Return
\end{eqnarray*}

The initial program state is $\left\langle \emptyset,\left\{ 0\mapsto\emptyset\right\} ,0,\mathsf{Void}\right\rangle $.
$0$ is a unique id referring to the global object.

\paragraph{Functions}

We redefine $\mathtt{enter}$ and $\mathtt{leave}$ functions to allow
passing of a receiver object.

\begin{eqnarray*}
\mathtt{enter} & = & \lambda\left\langle \_,\Omega,\_,\_\right\rangle ,n,P,T'.\left\langle \left[\mathtt{param}\left(n\right)_{k}\mapsto P_{k}\right],\Omega,T',\mathsf{Void}\right\rangle \\
\mathtt{leave} & = & \lambda\left\langle V,\_,T,\_\right\rangle ,\left\langle \_,\Omega',\_,r\right\rangle .\left\langle \left\langle V,\Omega',T,\mathsf{Void}\right\rangle ,r\right\rangle 
\end{eqnarray*}

The other functions are given in figure \ref{concfunc-oo}.

\begin{figure*}[!t]
\noindent\fbox{\begin{minipage}[t]{1\textwidth - 2\fboxsep - 2\fboxrule}%
\begin{eqnarray*}
\mathtt{get}\left(n,id\right) & = & \bar{\Theta}_{ObjMem}\ \left(\lambda\Omega.\Omega\left(n\right)\left(id\right)\right)\\
\mathtt{set}\left(n,id,v\right) & = & \bar{\Upsilon}_{ObjMem}\ \left(\lambda\Omega.\Omega\left[n\mapsto\Omega\left(n\right)\left[id\mapsto v\right]\right]\right)\\
\mathtt{getglobal} & = & \bar{\Theta}_{ObjMem}\ \left(\lambda\Omega.\Omega\left(0\right)\right)\\
\mathtt{getthis} & = & \bar{\Theta}_{This}\ \left(\lambda T.T\right)\\
\mathtt{apply}\left(\left\langle n,C\right\rangle ,p,t,\_\right) & = & \begin{cases}
\mathtt{call}\left(n,C\smallfrown p,t\right) & \mbox{if}\ \left\Vert C\smallfrown p\right\Vert =\mathtt{arity}\left(n\right)\\
\Gamma\ \left(\lambda\rho.\left\langle \rho,\left\langle n,C\smallfrown p\right\rangle \right\rangle \right) & \mbox{otherwise}
\end{cases}\\
\mathtt{newobj}\left(\_\right) & = & \bar{\Psi}_{ObjMem}\ \left(\lambda\Omega.\left\langle \Omega\left[n\mapsto\emptyset\right],n\right\rangle \right)\ \mbox{where}\ n=\left\Vert \Omega\right\Vert 
\end{eqnarray*}
\end{minipage}}

\caption{Concrete functions for the object-oriented extension}
\label{concfunc-oo}
\end{figure*}

\begin{example}[\label{Concrete-analysis-of}Concrete analysis of objects]
 Consider the following program.

\begin{lstlisting}
function Fruit(v) {
	this.value = v;
}

function juicible(fruit, juice) {
	function juiceMe(j,x) {
		return this.value + j + x;
	}
	fruit.juice = juiceMe(juice);
}

apple = new Fruit(15);
juicible(apple, 20);

output apple.juice(10); # 15 + 20 + 10
\end{lstlisting}
\end{example}

\begin{itemize}
\item Assume that $\mathtt{Fruit}$ has an id of 1, $\mathtt{juicible}$
has 2, and $\mathtt{juiceMe}$ has 3.
\item Right before line 12, we have $\left[\mathtt{Fruit}\mapsto\left\langle 1,\emptyset\right\rangle ,\mathtt{juicible}\mapsto\left\langle 2,\emptyset\right\rangle \right]$
. At line 12, the $\mathtt{new}$ expression creates an object in
object memory through $\mathtt{newobj}\left(\_\right)$, then passes
it on as a receiver of a method call to $\mathtt{Fruit}$. Inside
the \texttt{Fruit} function, the new object gets the member \texttt{value}.
After the line, we have $Env=\left[\mathtt{Fruit}\mapsto\left\langle 1,\emptyset\right\rangle ,\mathtt{juicible}\mapsto\left\langle 2,\emptyset\right\rangle ,\mathtt{apple}\mapsto\mathsf{Obj}\ 1\right]$
and $ObjMem=\left[1\mapsto\left[\mathtt{value}\mapsto15\right]\right]$.
\item At line 13, $\mathtt{apply}\left(\left\langle 2,\emptyset\right\rangle ,\left[\mathsf{Obj}\ 1,20\right],global,\_\right)$
gives $\mathtt{call}\left(2,\left[\mathsf{Obj}\ 1,20\right],global\right)$.
Inside the \texttt{juicible }function, object memory is manipulated
in line 9 to be $\left[1\mapsto\left[\mathtt{value}\mapsto15,\mathtt{juice}\mapsto\left\langle 3,\left[20\right]\right\rangle \right]\right]$.
\item At line 15, by the method call semantics, $\mathtt{apply}\left(\left\langle 3,\left[20\right]\right\rangle ,\left[10\right],\mathsf{Obj}\ 1,\_\right)$
is invoked, which gives 45 as a final result.
\end{itemize}

\subsubsection{Abstract interpretation}

Similarly, we extend the abstract interpretation. Note that, whereas
in a concrete interpretation we generate a unique id for each of the
objects instantiated at run-time, in abstract interpretation we use
a reference to the allocation site of an object as a reference to
a particular object, hence we contain object memory in a finite domain.

\paragraph{Domain}

\begin{eqnarray*}
AObj & = & \left\{ 0\right\} \cup Eid\\
AVal & = & \left\{ \mathsf{Num},\mathsf{Bool}\right\} \cup AObj\cup FunPointer\\
AObjMem & = & AObj\rightarrow AEnv\\
AThis & = & AObj\\
AState & = & AEnv\times AObjMem\times AThis\times Curried\times AReturn
\end{eqnarray*}

The initial program state is $\left\langle \emptyset,\left\{ 0\mapsto\emptyset\right\} ,0,\emptyset,\mathsf{Void}\right\rangle $.
$0$ is a unique id referring to the global object.

\paragraph{Functions}

\begin{eqnarray*}
\mathtt{enter} & = & \lambda\left\langle \_,\alpha,\_,\nu,\_\right\rangle ,n,P,\tau'.\left\langle \left[\mathtt{param}\left(n\right)_{k}\mapsto P_{k}\right],\alpha,\tau',\nu,\mathsf{Void}\right\rangle \\
\mathtt{leave} & = & \lambda\left\langle \sigma,\_,\tau,\_,\_\right\rangle ,\left\langle \_,\alpha',\_,\nu',r\right\rangle .\left\langle \left\langle \sigma,\alpha',\tau,\nu',\mathsf{Void}\right\rangle ,r\right\rangle 
\end{eqnarray*}

Other functions are defined in figure \ref{absfun-oo}.

\begin{figure*}[!t]
\noindent\fbox{\begin{minipage}[t]{1\textwidth - 2\fboxsep - 2\fboxrule}%
\begin{align*}
\mathtt{get}\left(n,id\right) & =\bar{\Theta}_{AObjMem}\ \lambda\alpha.\alpha\left(n\right)\\
\mathtt{set}\left(n,id,v\right) & =\bar{\Upsilon}_{AObjMem}\ \lambda\alpha.\alpha\left(n\right)\left[id\mapsto v\right]\\
\mathtt{obj}\left(j\right) & =\bar{\Theta}_{AObjMem}\ \lambda\alpha.\alpha\left(j\right)\\
\mathtt{getthis} & =\bar{\Theta}_{AThis}\ \lambda\tau.\tau\\
\mathtt{newobj}\left(n\right) & =\bar{\Psi}_{AObjMem}\ \lambda\alpha.\left\langle \alpha\left[n\mapsto\emptyset\right],n\right\rangle \\
\mathtt{apply}\left(\left\langle n,c,e\right\rangle ,p,t,e\right) & =\begin{cases}
\Theta_{{\scriptscriptstyle Curried}}\ \lambda\nu.\underset{C\in\nu\left(n,c,e\right)}{\bigcup}\mathtt{call}\left(n,C\smallfrown p,t\right)\ \mbox{if}\ c+\left\Vert p\right\Vert =\mathtt{arity}\left(n\right)\\
\Psi_{{\scriptscriptstyle Curried}}\ \lambda\nu.\underset{C\in\nu\left(n,c,e\right)}{\bigcup}\left\langle \nu\left[\left\langle n,c+\left\Vert p\right\Vert ,e\right\rangle \mapsto C\smallfrown p\right],\left\langle n,c+\left\Vert p\right\Vert ,e\right\rangle \right\rangle  & \mbox{otherwise}
\end{cases}
\end{align*}
\end{minipage}}

\caption{Abstract functions for the object-oriented extension}
\label{absfun-oo}
\end{figure*}

\begin{example}[Abstract interpretation of objects]
 The program in example \ref{Concrete-analysis-of} results in the
following final state.

\medskip{}

\begin{tabular}{|c|c|c|}
\hline 
Environment & Object Memory & Curried\tabularnewline
\hline 
\hline 
$\begin{aligned}\mathtt{Fruit} & \mapsto & \left\langle 1,0,0\right\rangle \\
\mathtt{juicible} & \mapsto & \left\langle 2,0,0\right\rangle \\
\mathtt{juiceMe} & \mapsto & \left\langle 3,0,0\right\rangle \\
\mathtt{apple} & \mapsto & \mathsf{Obj}\ 1
\end{aligned}
$ & $\mathsf{Obj}\ 1\mapsto\left(\begin{aligned}\mathtt{value} & \mapsto & \mathsf{\mathsf{Num}}\\
\mathtt{juice} & \mapsto & \left\langle 3,1,9\right\rangle 
\end{aligned}
\right)$ & $\left\langle 3,1,9\right\rangle \mapsto\left[\mathsf{\mathsf{Num}}\right]$\tabularnewline
\hline 
\end{tabular}
\end{example}

Notice that as we add object oriented programming to the programming
language, the functions are now capable of causing side-effects. The
following modified example illustrates how the meaning of recursive
function call is calculated when the function causes side-effects.
\begin{example}[Side-effect causing factorial function]
 Consider the following modification of program analysed in example
\ref{abs-int-fact}.

\begin{lstlisting}
function fact(f,n) {
	if(n>1) { global.x=5; return f(f,n-1) * n; } else { return 1; }
}

z=fact(fact,input);
output z;
\end{lstlisting}
\end{example}

\begin{itemize}
\item At line 5, the meaning of the function call is calculated as follows.
For simplicity, we only record $AEnv,AObjMem$ and $AReturn$. The
rest of the components remain unchanged.
\end{itemize}
\begin{tabular}{|c|c|c|}
\hline 
Current approximation & Meaning of function call & Note\tabularnewline
\hline 
\hline 
$\emptyset$ & $\left\{ \left\langle \begin{gathered}\left[\mathtt{fact}\mapsto\mathsf{Function}\ 1\right],\\
\left[0\mapsto\emptyset\right],\\
\mathsf{\mathsf{Num}}
\end{gathered}
\right\rangle \right\} $ & \tabularnewline
\hline 
$\left\{ \left\langle \begin{gathered}\left[\mathtt{fact}\mapsto\mathsf{Function}\ 1\right],\\
\left[0\mapsto\emptyset\right],\\
\mathsf{\mathsf{Num}}
\end{gathered}
\right\rangle \right\} $ & $\left\{ \begin{gathered}\left\langle \begin{gathered}\left[\mathtt{fact}\mapsto\mathsf{Function}\ 1\right],\\
\left[0\mapsto\emptyset\right],\\
\mathsf{\mathsf{Num}}
\end{gathered}
\right\rangle ,\\
\left\langle \begin{gathered}\left[\mathtt{fact}\mapsto\mathsf{Function}\ 1\right],\\
\left[0\mapsto\left[x\mapsto\mathsf{\mathsf{Num}}\right]\right],\\
\mathsf{\mathsf{Num}}
\end{gathered}
\right\rangle 
\end{gathered}
\right\} $ & \tabularnewline
\hline 
$\left\{ \begin{gathered}\left\langle \begin{gathered}\left[\mathtt{fact}\mapsto\mathsf{Function}\ 1\right],\\
\left[0\mapsto\emptyset\right],\\
\mathsf{\mathsf{Num}}
\end{gathered}
\right\rangle ,\\
\left\langle \begin{gathered}\left[\mathtt{fact}\mapsto\mathsf{Function}\ 1\right],\\
\left[0\mapsto\left[x\mapsto\mathsf{\mathsf{Num}}\right]\right],\\
\mathsf{\mathsf{Num}}
\end{gathered}
\right\rangle 
\end{gathered}
\right\} $ & $\left\{ \begin{gathered}\left\langle \begin{gathered}\left[\mathtt{fact}\mapsto\mathsf{Function}\ 1\right],\\
\left[0\mapsto\emptyset\right],\\
\mathsf{\mathsf{Num}}
\end{gathered}
\right\rangle ,\\
\left\langle \begin{gathered}\left[\mathtt{fact}\mapsto\mathsf{Function}\ 1\right],\\
\left[0\mapsto\left[x\mapsto\mathsf{\mathsf{Num}}\right]\right],\\
\mathsf{\mathsf{Num}}
\end{gathered}
\right\rangle 
\end{gathered}
\right\} $ & Fixed point\tabularnewline
\hline 
\end{tabular}

\subsection{\label{subsec:Exception-handling}Exception handling}

As a final extension to \sdtl, we introduce exception throwing and
handling. First, we introduce a familiar syntax for exception handling.

\begin{grammar}

<Stm> ::= `try' `{' <Stm> `}' `catch' `(' id `)' `{' <Stm> `}' | `throw' <Exp>

\end{grammar}

Following functions are the semantic functions for the newly introduced
syntax.

\begin{eqnarray*}
\mathcal{S}\dd{\mathtt{try}\ S_{1}\ \mathtt{catch}\left(id\right)\ S_{2}} & = & \mathcal{S}\ \dd{S_{1}}\bind_{noesc}\mathtt{catch}\left(id,\mathcal{S}\ \dd{S_{2}}\right)\\
\mathcal{S}\dd{\mathtt{throw}\ E} & = & \mathcal{E}\ \dd{E}\bind\lambda v.\mathtt{throw}\left(v\right)
\end{eqnarray*}

\begin{eqnarray*}
\mathtt{throw} & : & Value\rightarrow M\ Value\\
\mathtt{catch} & : & Id\rightarrow M\ Value\rightarrow M\ Value
\end{eqnarray*}

Note that we cannot use $\bind$ to combine the try code body and
catch, since the $\bind$ operator has $\mathtt{esc}$ built into
its definition, which will terminate the program flow when the exception
is raised. In order to allow exception-raising transformations to
flow through, we introduce a non-escaping bind operator, $\bind_{noesc}$.
$\mathtt{catch}$ primitive operator takes an exception variable name
and the meaning of a catch code block, and produces the meaning of
a language construct in which an exception is caught and handled.
\begin{defn}[Non-escaping bind operator]
 We define a non-escaping bind operator, $\bind_{noesc}$

\begin{eqnarray*}
\left(\bind_{noesc}\right) & : & M\ a\rightarrow\left(a\cup\left\{ \mathsf{Null}\right\} \rightarrow M\ a\right)\rightarrow M\ a\\
T\ \bind_{noesc}\ U & = & \lambda f\lambda s.\:\mathbf{let}\ S=T\ f\ s\ \mathbf{in}\underset{\left\langle s',a\right\rangle \in S}{\bigcup}\left(U\ a\ f\ s'\right)
\end{eqnarray*}

\pagebreak{}
\end{defn}

\subsubsection{Concrete interpretation}

As the control flow of the execution is abstracted out and parametrised
as a function $\mathtt{esc}$, it is not very surprising that introducing
exceptions to a program does not necessitate a heavy modification
to the interpretation of the language.

\paragraph{Domain}

\begin{eqnarray*}
Ex & = & Value\cup\left\{ \mathsf{Void}\right\} \\
CState & = & Env\times ObjMem\times This\times Return\times Ex
\end{eqnarray*}

The initial program state is $\left\langle \emptyset,\left\{ 0\mapsto\emptyset\right\} ,0,\mathsf{Void},\mathsf{Void}\right\rangle $.

\paragraph{Functions}

\begin{eqnarray*}
\mathtt{enter} & = & \lambda\left\langle \_,\Omega,\_,\_,\_\right\rangle ,\left\langle P,T'\right\rangle ,n.\left\langle \left[\mathtt{param}\left(n\right)_{k}\mapsto P_{k}\right],\Omega,T',\mathsf{Void},\mathsf{Void}\right\rangle \\
\mathtt{leave} & = & \lambda\left\langle V,\_,T,\_,\_\right\rangle ,\left\langle \_,\Omega',\_,r,e\right\rangle .\left\langle \left\langle V,\Omega',T,\mathsf{Void},e\right\rangle ,r\right\rangle 
\end{eqnarray*}

\begin{eqnarray*}
\mathtt{esc} & = & \Theta_{Return,Ex}\ \lambda R,E.\left(R\neq\mathsf{Void}\right)\wedge\left(E\neq\mathsf{Void}\right)\\
\mathtt{throw}\left(v\right) & = & \bar{\Upsilon}_{Ex}\ \lambda E.v\\
\mathtt{catch}\left(id,s\right) & = & \Theta_{Ex}\ \lambda E.\begin{cases}
I_{A}\mathsf{Unit} & \mbox{if}\ E=\mathsf{Void}\\
\lambda f,r.s\ f\ \mathtt{exs}\left(\rho,id\right) & \mbox{otherwise}
\end{cases}\\
\mathtt{exs}\left(id\right) & = & \Upsilon_{Env,Ex}\ \lambda V,E.\left\langle V\left[id\mapsto E\right],\mathsf{Void}\right\rangle 
\end{eqnarray*}

$\mathtt{exs}$ helps constructing an entering state to a catch block
if an exception was caught. As we do not support lexical scoping,
an exception variable enters a variable mapping as we enter the block.
\begin{example}[\label{Concrete-interpretation-of}Concrete interpretation of exception
handling]
 The following program illustrates exception handling in \sdtl.

\begin{lstlisting}
x = input;
try {
	if(x<0) {
		throw 0;
	}
	output x;
	j = 3;
} catch(e) {
	output x;
	output e;
}
\end{lstlisting}

\begin{itemize}
\item At line 3, we have $\left[\mathtt{x}\mapsto a\right]$ where $a$
is a user input. Suppose $a$ is less than 0, then \texttt{if} leads
to $\mathtt{throw}\left(0\right)$ which gives $Ex:0$.
\item The resulting state is deemed to be an escaping state by $\mathtt{esc}$.
Hence, the code block of the try clause gives the resulting program
state of $\left\langle \left[x\mapsto a\right],\left\{ 0\mapsto\emptyset\right\} ,0,\mathsf{Void},0\right\rangle $
\item With non-escaping binding to a $\mathtt{catch}$ primitive function
afterwards, \texttt{catch} yields the state with environment of $\left[x\mapsto a,e\mapsto0\right]$
then it evaluates the catch clause, which outputs $a$ and $0$.
\end{itemize}
\end{example}

It is important that this exception handling does not interfere with
other control flow. The following example shows how it interacts with
return state.
\begin{example}[Function call, return and exception handling control flow]
 Consider the following example.

\begin{lstlisting}
function tryorerror(func, error, a) {
	try {
		return func(a);
	} catch(e) {
		return error;
	}
}

function positive(a) {
	if(a>0) { return a; }
	throw 0;
}

gracefulpositive = tryorerror(positive, -1);

function doandprint(func,a) {
	output func(a);
}

try {
	doandprint(positive,50);
	doandprint(gracefulpositive,-50);
	doandprint(positive,-50);
} catch (e) {
	output e;
}
\end{lstlisting}

\begin{itemize}
\item At line 3, when $\mathtt{tryorerror}$ is called from line 21, it
sets a return state with \textsf{Void} exception state. As $Ex$ is
set to \textsf{Void}, the $\mathtt{catch}$ primitive operation does
not interfere, and the control flow proceeds to return the value of
50.
\item When called from line 22, the call to $\mathtt{positive}$ function
sets an exception state. Observe that, during the evaluation of the
return statement, if the exception is set while evaluating the returning
expression, the control flow does not proceed to set a return state.
the $\mathtt{catch}$ primitive operation notices exception state
being set, and proceeds to evaluate the exception handling block,
which returns -1. The meaning of the whole try-catch block, then,
is evaluated to be setting a return state of -1.
\item Finally, when the exception throwing function $\mathtt{positive}$
is directly called, at line 23, the exception state is set to be -50,
and $\mathtt{catch}$ primitive operation proceeds to line 25, printing
0 to the output.
\end{itemize}
\end{example}

\subsubsection{Abstract interpretation}

Abstract definitions are largely similar to that of concrete ones.
Notice that we carry on modified exception state returned from a function
call, thereby allowing exception propagation.

\paragraph{Domain}

\begin{eqnarray*}
AEx & = & AVal\cup\left\{ \mathsf{Void}\right\} \\
AState & = & AEnv\times AObjMem\times AThisCurried\times AReturn\times AEx
\end{eqnarray*}

The initial program state is $\left\langle \emptyset,\left\{ 0\mapsto\emptyset\right\} ,0,\emptyset,\mathsf{Void},\mathsf{Void}\right\rangle $.

\paragraph{Functions}

\begin{eqnarray*}
\mathtt{enter} & = & \lambda\left\langle \_,\alpha,\_,\nu,\_,\_\right\rangle ,\left\langle P,\tau'\right\rangle ,n.\left\langle \left[\mathtt{param}\left(n\right)_{k}\mapsto P_{k}\right],\alpha,\tau',\nu,\mathsf{Void},\mathsf{Void}\right\rangle \\
\mathtt{leave} & = & \lambda\left\langle \sigma,\_,\tau,\_,\_,\_\right\rangle ,\left\langle \_,\alpha',\_,\nu',r,e\right\rangle .\left\langle \left\langle \sigma,\alpha',\tau,\nu',\mathsf{Void},e\right\rangle ,r\right\rangle 
\end{eqnarray*}

\begin{eqnarray*}
\mathtt{esc} & = & \Theta_{AReturn,AEx}\ \lambda R,E.\left(R\neq\mathsf{Void}\right)\wedge\left(E\neq\mathsf{Void}\right)\\
\mathtt{throw}\left(v\right) & = & \bar{\Upsilon}_{AEx}\ \lambda E.v\\
\mathtt{catch}\left(id,s\right) & = & \Theta_{AEx}\ \lambda E.\begin{cases}
I_{A}\mathsf{Unit} & \mbox{if}\ E=\mathsf{Void}\\
\lambda f\lambda\eta.s\ f\ \mathtt{exs}\left(\eta,id\right) & \mbox{otherwise}
\end{cases}\\
\mathtt{exs}\left(id\right) & = & \Upsilon_{AEnv,AEx}\ \lambda\sigma,E.\left\langle \sigma\left[id\mapsto E\right],\mathsf{Void}\right\rangle 
\end{eqnarray*}

\begin{example}[Abstract interpretation of exception handling]
 The program in example \ref{Concrete-interpretation-of} can be
abstractly interpreted as follows.
\end{example}

\begin{itemize}
\item After lines 3-5, we have two possible states $\left\langle \left[\mathtt{x}\mapsto\mathsf{Num}\right],\left\{ 0\mapsto\emptyset\right\} ,0,\emptyset,\mathsf{Void},\mathsf{Void}\right\rangle $
and\\
 $\left\langle \left[\mathtt{x}\mapsto\mathsf{Num}\right],\left\{ 0\mapsto\emptyset\right\} ,0,\emptyset,\mathsf{Void},\mathsf{Num}\right\rangle $
\item By the bind operation, the latter state becomes the final state of
the try clause. The former state progresses further, and updates its
environment to $\left[\mathtt{x}\mapsto\mathsf{Num},\mathtt{j}\mapsto\mathsf{Num}\right]$
\item At line 8, the environment of exception throwing state is updated
to be $\left[\mathtt{x}\mapsto\mathsf{Num},\mathtt{e}\mapsto\mathsf{Num}\right]$,
then the catch clause is evaluated, which does not cause any state
update.
\item The final states are $\left\langle \left[\mathtt{x}\mapsto\mathsf{Num},\mathtt{j}\mapsto\mathsf{Num}\right],\left\{ 0\mapsto\emptyset\right\} ,0,\emptyset,\mathsf{Void},\mathsf{Void}\right\rangle $
and\\
 $\left\langle \left[\mathtt{x}\mapsto\mathsf{Num},\mathtt{e}\mapsto\mathsf{Num}\right],\left\{ 0\mapsto\emptyset\right\} ,0,\emptyset,\mathsf{Void},\mathsf{Void}\right\rangle $
\end{itemize}

\section{\label{sec:Well-definedness}Well-definedness}

All the functions appearing in both concrete and abstract versions
of interpretations are ordered point-wise. Recall that the bind operation
preserves the monotonicity when binding two monotonic functions. Observe
that all of the primitive operations that yield a state transformer,
and the value-returning operations lifted to yield one, are monotonic.
This follows from the fact that the domain of such transformers are
ordered by identity. Then, it follows that $\mathcal{S}$, $\mathcal{E}$
and $\mathcal{L}$ are monotonic for any given $f\in F$.

Now, observe the monotonic relationship between $F$ and monadic functions.
\begin{defn}[Point-wise ordering of function approximation]
Given $f,f'\in F$, then $f\sqsubseteq f'$ iff $\forall s\in Sid,\forall\sigma\in State,\left(f\ s\ \sigma\right)\subseteq\left(f'\ s\ \sigma\right)$
\end{defn}

\begin{thm}
$\forall f,f'\in F,f\sqsubseteq f'\implies\forall s\in\mathcal{S}\forall t\in Stm,\left(s\ t\ f\right)\sqsubseteq\left(s\ t\ f'\right)$
\end{thm}

\begin{proof}
Clearly, the $\mathtt{call}$ auxiliary function is monotonic with
regards to the function approximation. The $\mathtt{call}$ function
is the only function for which the function approximation value is
referenced. Recall the preservation of monotonicity of bind operation.
Then, this theorem immediately follows.
\end{proof}
It follows from this theorem that there exists a least fixed point
of the meaning of function $f\in F$ for a particular program. Hence
the well-definedness of the fixed point of $F$ in a particular analysis,
and that of the semantics both in concrete and abstract interpretations.

\section{\label{sec:Correctness}Correctness}

We define an abstraction relation $\succ$between concrete and abstract
states. 

\subsection{Definition of correct abstraction}

First, we define an abstraction between values and other components
of the states.
\begin{defn}
$\mathsf{Num}\succ\mathbb{Z}$, $\mathsf{\mathsf{Bool}}\succ\mathbb{B}$,
$\mathsf{Void}\succ\mathsf{Void}$

\begin{defn}[Abstraction of objects and symbol maps]
$\alpha,\Omega\vdash n\succ m$ if $\alpha\left(n\right)\succ\Omega\left(m\right)$

$\sigma:AEnv\succ V:Env$ if $\forall id\in dom\left(V\right),\sigma\left(id\right)\succ V\left(id\right)$

\begin{defn}[Function pointers]
$C:Curried\vdash\left\langle n,c,e\right\rangle \succ\left\langle m,P\right\rangle $
if $n=m$, $c=\left\Vert P\right\Vert $ \\
and $\exists A\in C\left(n,c,e\right),A_{k}\succ P_{k}$ for $1\leq k\leq c$
\end{defn}

\end{defn}

\end{defn}

Now we are ready to define an abstraction between concrete and abstract
states.
\begin{defn}[Abstraction relation]
Given

$\rho=\left(\begin{gathered}V:Env,\Omega:ObjMem,T:This,R:Return,E:Ex\end{gathered}
\right)\in CState$

$\eta=\left(\begin{gathered}\sigma:AEnv,\alpha:AObjMem,\tau:AThis,C:Curried,r:AReturn,e:AEx\end{gathered}
\right)\in AState$

$\eta\succ\rho$ if $\left(\sigma\succ V\right)\wedge\left(\forall m\in dom\left(\Omega\right)\exists n\in dom\left(\alpha\right),\alpha\left(n\right)\succ\Omega\left(m\right)\right)\wedge\left(\tau\succ T\wedge r\succ R\wedge e\succ E\right)$

\begin{defn}[Abstraction between powersets]
Given an abstract value $a$, a concrete value $c$,\\
 $p\in\wp\left(AState\times a\right)\mbox{ and }q\in\wp\left(CState\times c\right),p\succ q$
if

$\forall\rho\in q\exists\eta\in p,\left(\eta\succ\rho\right)\wedge\left(\rho,\eta\vdash a\succ c\right)$
\end{defn}

\end{defn}

\subsection{Morphism}

First, we observe that the bind operation preserves morphisms of the
functions
\begin{defn}
A pair of functions$\left\langle F,F'\right\rangle $ shows a morphism
over $\succ$ if the following condition is met

$a\succ b\implies F'\left(a\right)\succ F\left(b\right)$
\end{defn}

\begin{thm}
If $\left\langle F,F'\right\rangle $ and $\left\langle G,G'\right\rangle $
are pairs of functions that show a morphism over $\succ$, then $\left\langle F\bind G,F'\bind G'\right\rangle $
is also such pair.
\end{thm}

\begin{proof}
$\forall\rho\in CState$, let $p=F\left(\rho\right)$,$p'=F\bind G\left(\rho\right)$,
$\eta\in AState\mbox{ such that }\eta\succ\rho,$$q=F'\left(\eta\right)$,$q'=F'\bind G'\left(\eta\right)$.
\\
Then, $\forall p_{1}\in p,$

\begin{enumerate}
\item when $\mathtt{esc}\left(p_{1}\right)$ is true, then the morphism
and $\succ$ relation imply that $\exists q_{1}\in q,\mathtt{esc}\left(q_{1}\right)$.
Given the definition of $\bind$ operation, $p_{1}$ and $q_{1}$are
also elements of $p'$ and $q'$.
\item otherwise, morphism of $\left\langle G,G'\right\rangle $ gives that
$\forall p_{2}\in G\left(p_{1}\right)\exists q_{2}\in G'\left(q_{1}\right)\mbox{ such that }q_{2}\succ p_{2}$.
Such entities also exist in $q'$.
\end{enumerate}
\end{proof}
Finally, we observe that operations resulting from abstract functions
and their corresponding concrete functions form such a morphism over
$\succ$. We omit function-by-function proof, as such proof would
be a mechanical exercise.

\section{\label{sec:Conclusion-and-future}Conclusion and future direction}

As we have noted in the introduction, this work is a snapshot of an
ongoing dialogue between theory and practice, positively informing
each other to gradually move towards a better theorisation (and practical
implementation) of the difficult task of analysing dynamic languages.
We have sought to modularise the theoretical framework so that we
can take an incremental approach. We anticipate that, as a result
of having such a theory, adding new features to the current model
language while maintaining formality and rigour will be considerably
less laborious than to invent a new incarnation of a more feature-complete
model language and produce theory for it.

\bibliographystyle{IEEEtran}
\phantomsection\addcontentsline{toc}{section}{\refname}\bibliography{duck}

\pagebreak{}

\appendix

\section{Haskell Implementation}

\subsection{Syntax parser }

Creates an augmented syntax tree with unique id numbers.

\begin{lstlisting}[basicstyle={\footnotesize\ttfamily},language=Haskell,numbers=none,xleftmargin=0pt]
{
module Sdtl where
import Data.Char
import Control.Monad.State
import Data.Map

}

%name sdtl
%tokentype { Token }

%token
	if			{TokenIf}
	else		{TokenElse}
	while		{TokenWhile}
	input		{TokenInput}
	output		{TokenOutput}
	new			{TokenNew}
	this		{TokenThis}
	function	{TokenFunction}
	return		{TokenReturn}
	true		{TokenTrue}
	false		{TokenFalse}
	global		{TokenGlobal}
	try 		{TokenTry}
	catch 		{TokenCatch}
	throw  		{TokenThrow}
	id			{TokenId $$}
	'='			{TokenAsg}
	'+'			{TokenPlus}
	'-'			{TokenMinus}
	'*'			{TokenMult}
	'/'			{TokenDiv}
	gt			{TokenGreaterThan}
	lt			{TokenLessThan}
	eq			{TokenEqual}
	'('			{TokenLP}
	')'			{TokenRP}
	'{'			{TokenLC}
	'}'			{TokenRC}
	';'			{TokenScol}
	'.'			{TokenDot}
	','			{TokenComma}
	number		{TokenNumber $$}


%nonassoc gt lt eq
%left '+' '-'
%left '*' '/'

%monad {P}

%%


Program :: { SdtlProgram }
Program : Stmts {% state $ \(s0, m0) -> (SdtlProgram $1 m0, (s0,m0)) }

Stmts :: { AStmt }
Stmts : Stmts AStmt {% augStmt $ Stmts $1 $2}
	  | {% augStmt Empty}

AStmt :: { AStmt }
ASTmt : Stmt {% augStmt $1}      

Cblock :: { AStmt }
Cblock : '{' Stmts '}' { $2 }
	   | Stmt {% augStmt $1 }

Stmt :: {Stmt}
Stmt : AExpr ';' {Expr $1}
	 | Lexpr '=' AExpr ';' {Asg $1 $3}
	 | if '(' AExpr ')' Cblock else Cblock {Ite $3 $5 $7}
	 | if '(' AExpr ')' Cblock {If $3 $5}
	 | while '(' AExpr ')' Cblock {While $3 $5}
	 | output AExpr ';' {Output $2}
	 | return AExpr ';' {Return $2}
	 | function id '(' Idlist ')' Cblock {Function $2 $4 $6}
	 | try Cblock catch '(' id ')' Cblock {TryCatch $2 $5 $7}
	 | throw AExpr ';' {Throw $2}

Params :: {[AExpr]}
Params : {[]}
	   | AExpr {[$1]}
	   | AExpr ',' Params {$1:$3}

Idlist :: {[String]}
Idlist : {[]}
	   | id {[$1]}
	   | id ',' Idlist {$1:$3}

AExpr :: {AExpr}
AExpr : Expr {% augExpr $1}

Expr :: {Expr}
Expr : Lexpr {Lexpr $1}
	 | true {Boolean True}
	 | false {Boolean False}
	 | number {Number $1}
	 | AExpr '+' AExpr {Binop Plus $1 $3}
	 | AExpr '-' AExpr {Binop Minus $1 $3}
	 | AExpr '*' AExpr {Binop Mult $1 $3}
	 | AExpr '/' AExpr {Binop Div $1 $3}
	 | AExpr gt AExpr {Binop GreaterThan $1 $3}
	 | AExpr lt AExpr {Binop LessThan $1 $3}
	 | AExpr eq AExpr {Binop Equal $1 $3}
	 | input {Input}
	 | new AExpr {New $2}
	 | Lexpr '(' Params ')' {Call $1 $3}
	 | '(' Expr ')' {$2}
	 | this {This}
	 | global {Global}
Lexpr :: {Lexpr}
Lexpr : '(' Lexpr ')' {$2}
	  | AExpr '.' id {Ref $1 $3}
	  | id {Id $1}
{

type Sid = Int
type Eid = Int

data SdtlProgram = SdtlProgram AStmt (Map Int AStmt) deriving (Show)

data Lexpr = 
	Ref AExpr String
	| Id String
	deriving Show

data Op = Plus | Minus | Mult | Div | GreaterThan | LessThan | Equal
	deriving Show

data AExpr = AExpr Expr Eid deriving Show

data Expr =
	Lexpr Lexpr
	| Boolean Bool
	| Number Int
	| Binop Op AExpr AExpr
	| Input
	| New AExpr
	| Call Lexpr [AExpr]
	| This
	| Global
	deriving Show

data AStmt = 
	AStmt Stmt Sid 
	deriving Show

data Stmt =
	Asg Lexpr AExpr
	| Ite AExpr AStmt AStmt
	| If AExpr AStmt
	| While AExpr AStmt
	| Output AExpr
	| Expr AExpr
	| Return AExpr
	| Function String [String] AStmt
	| TryCatch AStmt String AStmt
	| Empty
	| Stmts AStmt AStmt
	| Throw AExpr
	deriving Show

data Token =
	TokenIf
	| TokenElse
	| TokenWhile
	| TokenInput
	| TokenOutput
	| TokenNew
	| TokenThis
	| TokenFunction
	| TokenReturn
	| TokenTrue
	| TokenFalse
	| TokenGlobal
	| TokenId String
	| TokenAsg
	| TokenPlus
	| TokenMinus
	| TokenMult
	| TokenDiv
	| TokenGreaterThan
	| TokenLessThan
	| TokenEqual
	| TokenLP
	| TokenRP
	| TokenLC
	| TokenRC
	| TokenScol
	| TokenDot
	| TokenComma
	| TokenNumber Int
	| TokenTry
	| TokenCatch
	| TokenThrow
	deriving Show

type P = State (Int, Map Int AStmt)

augStmt :: Stmt -> P AStmt
augStmt s = state $ \(s0, m0) -> let news = AStmt s s0 in (news, (s0+1, insert s0 news m0 ))

augExpr :: Expr -> P AExpr
augExpr e = state $ \(s0, m0) -> (AExpr e s0, (s0+1, m0))

lexer :: String -> [Token]
lexer s = lexer1 s

lexer1 :: String -> [Token]

lexer1 [] = []
lexer1 (c:cs)
	| isSpace c = lexer1 cs
	| c == '\t' = lexer1 cs
	| c == '\r' = lexer1 cs
	| c == '\n' = lexer1 cs
	| isAlpha c = lexVar (c:cs)
	| isDigit c = lexNum (c:cs)
	| c == '-' = lexNegativeNum (cs)
	| c == '#' = lexer1 (dropWhile (/= '\n') cs)

lexer1 ('=':cs) = 
	case cs of
		('=':ss) -> TokenEqual : lexer1 ss
		_ -> TokenAsg : lexer1 cs

lexer1 ('+':cs) = TokenPlus : lexer1 cs
lexer1 ('-':cs) = TokenMinus : lexer1 cs
lexer1 ('*':cs) = TokenMult : lexer1 cs
lexer1 ('/':cs) = TokenDiv : lexer1 cs
lexer1 ('(':cs) = TokenLP : lexer1 cs
lexer1 (')':cs) = TokenRP : lexer1 cs
lexer1 ('{':cs) = TokenLC : lexer1 cs
lexer1 ('}':cs) = TokenRC : lexer1 cs
lexer1 ('.':cs) = TokenDot : lexer1 cs
lexer1 (';':cs) = TokenScol : lexer1 cs
lexer1 ('>':cs) = TokenGreaterThan : lexer1 cs
lexer1 ('<':cs) = TokenLessThan : lexer1 cs
lexer1 (',':cs) = TokenComma : lexer1 cs

lexNum cs = TokenNumber (read num) : lexer1 rest 
	where (num,rest) = span isDigit cs

lexNegativeNum cs = TokenNumber (-1 * (read num)) : lexer1 rest 
	where (num,rest) = span isDigit cs

lexVar cs =
	case span isAlphaNum cs of
		("if", rest) -> TokenIf : lexer1 rest 
		("else", rest) -> TokenElse : lexer1 rest
		("while", rest) -> TokenWhile : lexer1 rest
		("input", rest) -> TokenInput : lexer1 rest
		("output", rest) -> TokenOutput : lexer1 rest
		("new", rest) -> TokenNew : lexer1 rest
		("this", rest) -> TokenThis : lexer1 rest
		("function", rest) -> TokenFunction : lexer1 rest
		("return", rest) -> TokenReturn : lexer1 rest
		("true", rest) -> TokenTrue : lexer1 rest
		("false", rest) -> TokenFalse : lexer1 rest
		("global", rest) -> TokenGlobal : lexer1 rest
		("try", rest) -> TokenTry : lexer1 rest
		("catch", rest) -> TokenCatch : lexer1 rest
		("throw", rest) -> TokenThrow : lexer1 rest
		(var, rest) -> TokenId var : lexer1 rest

happyError :: [Token] -> a
happyError _ = error "Parse error"
}
\end{lstlisting}

\subsection{Semantic functions}

\begin{lstlisting}[basicstyle={\footnotesize\ttfamily},language=Haskell,numbers=none,xleftmargin=0pt]
{-# LANGUAGE GADTs, MultiParamTypeClasses #-}
module SemanticFunctions where
import Data.Maybe as Maybe
import Sdtl
import Data.Map as Map
import Debug.Trace

class (Eq v) => Value v where
	conval :: Expr -> v
	getSid :: v -> Sid
	binop :: Op -> v -> v -> v
class (Eq s) => State s where
	fundecl :: String -> [String] -> Sid -> s -> [s]
	esc	:: s -> Bool
	
class (State s, Value v) => StateValue s v where
	ret :: v -> s -> [s]
	cond :: (Eq a) => v -> (M s v a) -> (M s v a) -> (M s v a)
	asg :: String -> v -> s -> [s]
	dooutput :: v -> s -> [s]
	getinput :: s -> [(s,v)]
	val :: s -> String -> v
	enter :: s -> Sid -> [v] -> v -> [String] -> s
	leave :: s -> s -> (s, Maybe v)
	fix :: ((M s v ()) -> (M s v ())) -> M s v ()
	apply :: v -> [v] -> v -> Eid -> (M s v v)
	set :: v -> String -> v -> s -> [s]
	get :: v -> String -> s -> v
	getglobal :: s -> v
	getthis :: s -> v
	newobj :: Eid -> (M s v v)
	runthrow :: v -> s -> [s]
	runcatch :: String -> (M s v ()) -> (M s v ())
data M s v a where
	M :: (State s, Value v) => ((Sid -> (AStmt, (M s v ()))) -> s -> [(s,Maybe a)]) -> M s v a

--data (State s) => M s a = M (F s -> s -> [(s,Maybe a)])

instance (StateValue s v) => Monad (M s v) where
	(M t) >>= u = 
		M $ (\f -> \s0 -> 
			let ss = t f s0
			in
			   concat [ if esc s1 then [(s1, Nothing)]
			   	            	  else let 
			   	            	      Just a' = a
			   	            	      (M us) = u a'
			   	            	      in us f s1
			   	      | (s1, a) <- ss ])
	return a = M $ \f -> \s0 -> [(s0, Just a)]

(>>>) :: (StateValue s v) => (M s v a) -> (Maybe a -> (M s v a)) -> (M s v a)
(M t) >>> u = M $ (\f -> \s0 -> 
	let ss = t f s0 in concat [ let (M us) = u a in us f s1 | (s1, a) <- ss ])


{- Takes a state transformer and transforms into a monad -}
liftS :: (StateValue s v) => (s -> [s]) -> M s v ()
liftS strans = M $ \f -> \s0 -> [(s1,Just ()) | s1 <- (strans s0)]

liftV :: (State s, Value v) => (s -> v) -> M s v v
liftV stov = M $ \f -> \s -> [(s, Just (stov s))]

getState :: (State s, Value v) => M s v s
getState = M $ \f -> \s0 -> [(s0,Just s0)]

sstmt :: (StateValue s v) => AStmt -> M s v ()
sstmt (AStmt Empty sid) = return ()
sstmt (AStmt (Stmts s1 s2) sid)  = do 
	traceShow s1 $ sstmt s1
	traceShow s2 $ sstmt s2
sstmt (AStmt (Expr e1) sid) = do
	sexpr e1
	return ()
sstmt (AStmt (Return e1) sid) = do
	v <- sexpr e1
	liftS $ ret v
sstmt (AStmt (If e1 s1) sid) = do
	v <- sexpr e1
	cond v (sstmt s1) (return ())
sstmt (AStmt (Ite e1 s1 s2) sid) = do
	v <- sexpr e1
	cond v (sstmt s1) (sstmt s2)
sstmt (AStmt (Asg (Id i1) e1) sid) = do
	v <- sexpr e1
	liftS $ asg i1 v
sstmt (AStmt (While e1 s1) sid) = fix $ \x -> do
	v <- sexpr e1
	cond v (sstmt s1 >> x) (return ())
sstmt (AStmt (Output e1) sid) = do
	v <- sexpr e1
	liftS $ dooutput v
sstmt (AStmt (Function id1 ids (AStmt _ s1)) sid) = liftS $ fundecl id1 ids sid
sstmt (AStmt (Asg (Ref e1 id1) e2) sid) = do
	r <- sexpr e1
	v <- sexpr e2
	liftS $ set r id1 v
sstmt (AStmt (TryCatch s1 id1 s2) sid) =
	(sstmt s1) >>> \_ -> runcatch id1 (sstmt s2)
sstmt (AStmt (Throw e1) sid) = do
	v1 <- sexpr e1
	liftS $ runthrow v1

sexpr :: (StateValue s v) => AExpr -> M s v v
sexpr (AExpr (Number n) eid) = return $ conval (Number n)
sexpr (AExpr (Boolean b) eid) = return $ conval (Boolean b)
sexpr (AExpr (Lexpr l1) eid) = slexpr l1
sexpr (AExpr Input eid) = M $ \f -> \s0 -> [(s, Just v) | (s,v) <- getinput s0]
sexpr (AExpr (Call (Ref e1 id1) es) eid) = do
	t <- sexpr e1
	n <- liftV $ get t id1
	p <- evalParams es []
	apply n p t eid
sexpr (AExpr (Call l1 es) eid) = do
	n <- slexpr l1
	p <- evalParams es []
	t <- liftV $ getthis
	apply n p t eid
sexpr (AExpr (Binop op e1 e2) eid) = do
	v1 <- sexpr e1
	v2 <- sexpr e2
	return $ binop op v1 v2
sexpr (AExpr This eid) = liftV $ getthis
sexpr (AExpr Global eid) = liftV $ getglobal
sexpr (AExpr (New (AExpr (Call l1 es) eid1)) eid) = do
	n <- slexpr l1
	p <- evalParams es []
	m <- newobj eid
	apply n p m eid
	return m
evalParams :: (StateValue s v) => [AExpr] -> [v] -> M s v [v]
evalParams [] vs = return vs
evalParams (e:es) vs = do
	v <- sexpr e
	evalParams es (vs ++ [v])

slexpr :: (StateValue s v) => Lexpr -> M s v v
slexpr (Id id1) = do
	s0 <- getState 
	return $ val s0 id1
slexpr (Ref e1 id1) = do
	v <- sexpr e1
	liftV $ get v id1
--putState :: (StateValue s v) => s -> (M s v ())
--putState s = M $ \f -> \s0 -> [(s,Just ())]

--getFunc :: (StateValue s v) => M s v (Sid -> (M s v ()))
--getFunc = M $ \f -> \s -> [(s, Just f)]

call :: (StateValue s v) => Sid -> [v] -> v -> (M s v v)
call n p t = M $ \f -> \s0 -> 
	[leave s0 s | (s,_) <- let (AStmt (Function _ ids _) _, M x) = (f n) in x f (enter s0 n p t ids) ]

runMonad :: (StateValue s v) => AStmt -> (Sid -> (AStmt, (M s v ()))) -> s -> v -> [s]
runMonad mainStmt = \f -> \s0 -> \_ -> let M t = sstmt mainStmt in (Prelude.map) fst (t f s0)

params :: (Sid -> (AStmt, (M s v ()))) -> Sid -> [String]
params f sid = let ((AStmt (Function _ ids _) _), _) = (f sid) in ids

arity :: (Sid -> (AStmt, (M s v ()))) -> Sid -> Int
arity f sid = length (params f sid)
\end{lstlisting}

\subsection{Concrete interpretation}

\begin{lstlisting}[basicstyle={\footnotesize\ttfamily},language=Haskell,numbers=none,xleftmargin=0pt]
{-# LANGUAGE MagicHash, UnboxedTuples, GADTs, MultiParamTypeClasses #-}

module Main (main, runString) where

import SemanticFunctions
import Sdtl
import Data.Map
import qualified Control.Exception as E
import GHC.IO
import GHC.Exts
import System.Environment
import qualified Control.Monad.State as S
import System.IO.Unsafe
import Data.ByteString.Internal
import Debug.Trace
{-
fromState :: State# RealWorld -> RealWorld
fromState = unsafeCoerce#

toState :: RealWorld -> State# RealWorld
toState = unsafeCoerce#

execIO :: IO a -> RealWorld -> (RealWorld, a)
execIO (IO k) = \w -> case k (toState w) of (# s', x #) -> (fromState s', x)
-}

data CValue = Intval Int
            | Boolval Bool
            | FunPointer Sid [CValue]
            | Object Int
            deriving (Show, Eq)

instance Value CValue where
	conval (Boolean True) = Boolval True
	conval (Boolean False) = Boolval False
	conval (Number i) = Intval i
	getSid (FunPointer n []) = n
	binop Plus (Intval i1) (Intval i2) = Intval $ i1 + i2
	binop Minus (Intval i1) (Intval i2) = Intval $ i1 - i2
	binop Mult (Intval i1) (Intval i2) = Intval $ i1 * i2
	binop Div (Intval i1) (Intval i2) = Intval $ i1 `div` i2
	binop GreaterThan (Intval i1) (Intval i2) = Boolval $ i1 > i2
	binop LessThan (Intval i1) (Intval i2) = Boolval $ i1 < i2
	binop Equal (Intval i1) (Intval i2) = Boolval $ i1 == i2
	
data CState = CState{
		env :: Map String CValue,
		--io :: RealWorld,
		objmem :: Map Int (Map String CValue),
		returning :: Maybe CValue,
		exception :: Maybe CValue,
		this :: CValue
	} deriving (Show)

instance Eq CState where
	(==) _ _ = False

instance State (CState) where
	esc (CState {returning = v, exception = e}) = case v of
		Nothing -> case e of 
			Nothing -> False
			_ -> True
		_ -> True
	
	--fundecl :: String -> [String] -> Sid -> s -> [s]
	fundecl fnname fnparams n s = 
		[s {env = insert fnname (FunPointer n []) (env s) }]
instance StateValue CState CValue where
	ret v s = trace (show s) $ [s {returning = Just v}]

	--cond :: (CValue v) => v -> (M s v a) -> (M s v a) -> (M s v a)
	cond (Boolval True) t u = t
	cond _ t u = u

	--asg :: (CValue v) => String -> v -> s -> [s]
	asg i v s = [s {env = (insert i v (env s))}]

	--dooutput :: (CValue v) => v -> s -> [s]
	dooutput (Intval n) s =
		--let (io1, _) = (execIO $ putStrLn (show n)) (io s)
		--in [s {io = io1}]
		let v = ($!) inlinePerformIO $ putStrLn (show n) in
		if v == () then [s] else [s]
	
	--getinput :: (CValue v) => s -> [(s,v)]
	getinput s =
		--let (io1, v) = (execIO $ getLine) (io s)
		--in [(s {io = io1}, Intval (read v))]
		[(s,Intval (read (inlinePerformIO $ getLine)))]

	--val :: (CValue v) => s -> String -> v
	val s i = (env s) ! i

	--enter :: (CValue v) => s -> Sid -> [v] -> s
	enter s n vs t ids = s {env = fromList (zip ids vs), this = t}

	--leave :: (CValue v) => s -> s -> (s, Maybe v)
	leave s0 s = trace (show s) $ 
		(CState {env = (env s0), {-io = (io s), -} 
			     returning = Nothing,
			     exception = (exception s),
			     objmem = (objmem s),
			     this = (this s0)}, (returning s))

	fix x = x (fix x)

	apply (FunPointer n c) p t _ =
		M $ \f -> \s -> 
			if (arity f n) == (length (c ++ p)) 
				then let (M tcall) = call n (c ++ p) t in (tcall f s)
				else [(s, Just (FunPointer n (c++p)))]

	set (Object n) id1 v = \s -> 
		let oldo = (objmem s) ! n
		in [s {objmem = insert n (insert id1 v oldo) (objmem s)}]

	get (Object n) id1 = \s -> (objmem s) ! n ! id1
	
	getglobal = \s -> (Object 0)
	getthis = \s -> (this s)
	newobj _ = M $ \f -> \s -> 
		let newn = size (objmem s)
		in 
		traceShow (Object newn) $
		[(s {objmem = insert newn (fromList []) (objmem s)}, Just (Object newn))]

	--runthrow :: v -> s -> [s]
	runthrow v = \s -> [s {exception=Just v}]

	--runcatch :: String -> (M s v ()) -> (M s v ())
	runcatch id1 (M u) = M $ \f -> \s0 ->
		case (exception s0) of
			Nothing -> [(s0, Just ())]
			Just v -> (u f (s0 {env = insert id1 v (env s0), exception=Nothing}))
{-
getWorld :: IO RealWorld
getWorld = IO (\s -> (# s, fromState s #))

putWorld :: RealWorld -> IO ()
putWorld s' = IO (\_ -> (# toState s', () #))
-}


fun :: (StateValue s v) => (Map Int AStmt) -> Int -> (AStmt, (M s v ()))
fun sidMap sid = trace ("call sid " ++ (show sid)) $ 
	let Just astmt = Data.Map.lookup sid sidMap
	    (AStmt (Function _ _ fbody) _) = astmt in (astmt, sstmt fbody)

runConcrete :: SdtlProgram -> IO ()
runConcrete (SdtlProgram mainStmt sdtlMap) = do
	--putStrLn $ show mainSid
	putStrLn $ show sdtlMap
	--w <- getWorld
	let [endState] = 
		runMonad mainStmt
				 (fun sdtlMap)
		         (CState {env = fromList [], {-io = w, -} returning = Nothing, exception = Nothing, 
		         		  objmem = fromList [(0, fromList [])], this = (Object 0)}) 
		         (Intval 0)
		in do
		--putWorld (io endState)
		putStrLn $ show (env endState)

runString :: String -> IO()
runString program = runConcrete (fst (S.runState (sdtl (lexer program)) (1, fromList [])))

main = do
	(sdtlfile:_) <- getArgs
	program <- readFile sdtlfile
	runString program
\end{lstlisting}

\subsection{\label{subsec:Abstract-interpretation}Abstract interpretation}

\begin{lstlisting}[basicstyle={\footnotesize\ttfamily},language=Haskell,numbers=none,xleftmargin=0pt]
{-# LANGUAGE MultiParamTypeClasses #-}
module Main (main, runString) where

import SemanticFunctions
import Sdtl
import Data.Map
import qualified Control.Exception as E
import System.Environment
import Data.ByteString.Internal
import Debug.Trace
import qualified Control.Monad.State as S
import qualified Data.List as L

data AValue = Intval
            | Boolval
            | FunPointer Sid Int Eid
            | Object Int
            deriving (Show, Eq, Ord)

instance Value AValue where
    conval (Boolean True) = Boolval
    conval (Boolean False) = Boolval
    conval (Number i) = Intval
    getSid (FunPointer n _ _) = n
    binop Plus (Intval) (Intval) = Intval
    binop Minus (Intval) (Intval) = Intval
    binop Mult (Intval) (Intval) = Intval
    binop Div (Intval) (Intval) = Intval
    binop GreaterThan (Intval) (Intval) = Boolval
    binop LessThan (Intval) (Intval) = Boolval
    binop Equal (Intval) (Intval) = Boolval
    
data AState = AState{
        env :: Map String AValue,
        returning :: Maybe AValue,
        exception :: Maybe AValue,
        objmem :: Map Eid (Map String AValue),
        this :: AValue,
        curried :: Map (Sid, Int, Eid) [AValue]
    } deriving (Show, Eq, Ord)

instance State (AState) where
    esc (AState {returning = v, exception = e}) = case v of
        Nothing -> case e of 
            Nothing -> False
            _ -> True
        _ -> True
    
    --fundecl :: String -> [String] -> Sid -> s -> [s]
    fundecl fnname fnparams n s = 
        [s {env = insert fnname (FunPointer n 0 0) (env s) }]

y x = x (y x)
mapToMonad m = M $ \f -> \s -> if member s m then (m ! s) else [(s, Just ())]

instance StateValue AState AValue where
    fix x = M $ \f -> \s ->
        let (fixedMap) = 
                (y $ \x0 -> \m0 -> 
                let (M t1) = (x (mapToMonad m0))
                    newmap = traceShow m0 $ 
                             union m0 $ 
                             fromList 
                             [(s1,t1 f s1) | s1 <- (L.map fst (L.nub $ concat $ elems m0)) ++ [s] ]
                in
                if m0 == newmap then newmap else x0 newmap) (fromList [])
        in
            trace ("final " ++ (show fixedMap)) $
                L.nub $ concat $ elems fixedMap

    ret v s = trace (show s) $ [s {returning = Just v}]

    --cond :: (CValue v) => v -> (M s v a) -> (M s v a) -> (M s v a)
    cond (Boolval) (M t) (M u) = M $ \f -> \s0 ->
        let s1 = t f s0
            s2 = u f s0
        in L.nub $ s1 ++ s2
    --cond _ t u = u

    --asg :: (CValue v) => String -> v -> s -> [s]
    asg i v s = [s {env = (insert i v (env s))}]

    --dooutput :: (CValue v) => v -> s -> [s]
    dooutput _ s = [s]

    --getinput :: (CValue v) => s -> [(s,v)]
    getinput s =
        [(s,Intval)]

    --val :: (CValue v) => s -> String -> v
    val s i = (env s) ! i

    --enter :: (CValue v) => s -> Sid -> [v] -> s
    enter s n vs t ids = s {env = fromList (zip ids vs), this = t}

    --leave :: (CValue v) => s -> s -> (s, Maybe v)
    leave s0 s = trace (show s) $ 
        (AState {env = (env s0), returning = Nothing, exception = (exception s), 
                curried = (curried s),
                objmem = (objmem s),
                this = (this s0)}, (returning s))

    apply (FunPointer n c ce) p t e =
        M $ \f -> \s ->
            let cv=if c > 0 then ((curried s) ! (n, c, ce)) else [] in
                if (arity f n) == c + (length p) 
                    then let (M tcall) = (call n (cv ++ p) t) in (tcall f s)
                    else 
                        [(s {curried = insert (n, (c + length p), e)  (cv ++ p) (curried s)}
                            , Just (FunPointer n (c+(length p)) e))]

    set (Object n) id1 v = \s -> 
        let oldo = (objmem s) ! n
        in [s {objmem = insert n (insert id1 v oldo) (objmem s)}]

    get (Object n) id1 = \s -> (objmem s) ! n ! id1
    
    getglobal = \s -> (Object 0)
    getthis = \s -> (this s)
    newobj eid = M $ \f -> \s -> 
        [(s {objmem = insert eid (fromList []) (objmem s)}, Just (Object eid))]

    --runthrow :: v -> s -> [s]
    runthrow v = \s -> [s {exception=Just v}]

    --runcatch :: String -> (M s v ()) -> (M s v ())
    runcatch id1 (M u) = M $ \f -> \s0 ->
        case (exception s0) of
            Nothing -> [(s0, Just ())]
            Just v -> (u f (s0 {env = insert id1 v (env s0), exception=Nothing}))
{-
getWorld :: IO RealWorld
getWorld = IO (\s -> (# s, fromState s #))

putWorld :: RealWorld -> IO ()
putWorld s' = IO (\_ -> (# toState s', () #))
-}

fun :: (Map Int AStmt) -> (Map (Int,AState) [(AState, Maybe ())]) -> Int 
    -> (AStmt, (M AState AValue ()))
fun sidMap fapprox sid = trace ("call sid " ++ (show sid)) $ 
    let Just astmt = Data.Map.lookup sid sidMap
        (AStmt (Function _ _ fbody) _) = traceShow astmt $ astmt in
    (astmt,
        (M $ \f -> \s -> 
                if member (sid, s) fapprox then fapprox ! (sid, s)
                else  
                (y $ \localx -> \prevresult ->
                    let 
                        newapprox = (insert (sid, s) prevresult fapprox)
                        M t = (sstmt fbody)
                        result = t (fun sidMap newapprox) s
                    in
                        if result == prevresult then result
                        else localx result) []))

runAbstract :: SdtlProgram -> IO ()
runAbstract (SdtlProgram mainStmt sdtlMap) = do
    --putStrLn $ show mainSid
    putStrLn $ show sdtlMap
    --w <- getWorld
    let endStates = 
            runMonad mainStmt
                     (fun sdtlMap (fromList []))
                     (AState {env = fromList [], 
                              returning = Nothing,
                              exception = Nothing,
                              curried = fromList [],
                              objmem = fromList [(0, fromList [])],
                              this = (Object 0)}) 
                     (Intval)
        in do
        --putWorld (io endState)
        putStrLn "Final States"
        putStrLn $ show endStates

runString :: String -> IO()
runString program = runAbstract (fst (S.runState (sdtl (lexer program)) (1, fromList [])))

main = do
    (sdtlfile:_) <- getArgs
    program <- readFile sdtlfile
    runString program
\end{lstlisting}
\end{document}